%% file: main.tex
\pdfoutput=1
\documentclass[manyauthors]{fundam}

\usepackage{geometry}
\usepackage{core}
\usepackage{booktabs}
\usepackage{multirow}
\usepackage{hcalculation} 

\usepackage{tikz/tikzit}
\usepackage{tikz/gla}
\input{tikz/gla.tikzstyles}

\mathcode`;=\numexpr\mathcode`;-"6000

\DeclareMathOperator{\im}{Im}
\let\ker\relax
\DeclareMathOperator{\ker}{Ker}
\def\by{{\times}}
\def\K{\mathbb{K}}

\usetikzlibrary{external}
\tikzset{external/only named=true}
\title{Generalizing the Invertible Matrix Theorem with Linear Relations using Graphical Linear Algebra}
\date{\today}

\author{%
Iago Leal de Freitas%
\thanks{The authors would like to thank Gustavo Freire and David E. Bernal Neira for their help reviewing this manuscript.}%
\\ Purdue University \\ West Lafayette, IN 47907, USA \\
\and
Júlia Mota\thanksas{1}%
\thanks{This study was financed in part by the Coordenação de Aperfeiçoamento de Pessoal de Nível Superior - Brasil (CAPES) - Finance Code 001.}%
\\ Universidade Federal do Rio de Janeiro \\ Rio de Janeiro, RJ, Brazil
\and
João Paixão\thanksas{1}%
\\ Universidade Federal do Rio de Janeiro \\ Rio de Janeiro, RJ, Brazil
\and
Lucas Rufino\thanksas{1}\thanksas{2}\corresponding%
\\ Instituto Nacional de Matemática Pura e Aplicada \\ Rio de Janeiro, RJ, Brazil
}

\runninghead{I. Leal de Freitas, J. Mota, J. Paixão, L. Rufino}{A Generalized IMT via Relational Decompositions}

\begin{document}

\maketitle

\begin{abstract}
Linear algebra's main concerns are sets of vectors, linear functions, subspaces, linear systems, matrices and concepts about those, such as whether the solution of linear system exists or is unique; a set of vectors is linearly independent or spans the whole space; a linear function has a right or a left inverse; a linear function is surjective or injective; and the kernel of a matrix is trivial or the its image is full.

The Invertible Matrix Theorem ties all these ideas and many others together.
Many modern linear algebra books use this theorem as a guiding principle to explain many connections in linear algebra. The main idea is to separately characterize whether the linear function is surjective or injective. The proof usually uses a matrix decomposition as the key step.
However, the invertible matrix theorem deals with a single linear function, a single set of vectors, a single subspace, and a single matrix.

In this work, we generalize part of the invertible matrix theorem to results about
a pair of linear functions, a pair of sets of vectors, a pair of subspaces, and a single linear relation.
The main idea is to separately characterize the linear relation's fundamental properties---whether it is surjective, injective, deterministic and total.
Our proof uses a decomposition of a linear relation as the key step.

Unfortunately, reasoning with linear relations in classical notation requires
applying many rules besides shuffling quantifiers and variables around,
which can obscure the symmetries in the results.
Therefore,  this work employs graphical linear algebra,
a two-dimensional diagrammatic syntax with the fundamental rules of linear relations built-in.
\end{abstract}

\begin{keywords}
Linear Relations, Matrix Decomposition, Graphical Linear Algebra, Invertible Matrix Theorem, String Diagrams
\end{keywords}

\tableofcontents
\newpage

\section{Introduction}
\label{sec:introduction}

In linear algebra, there are many ways of stating that a matrix is invertible. For example, a $n$-by-$n$ matrix $A$ is invertible when either its kernel is zero, its image is full, or it admits an inverse $AB = I = BA$. The central result that establishes all these notions of invertibility as equivalent is the Invertible Matrix Theorem~\cite{weisstein2014invertible}. Its only underlying hypothesis is that the matrix A must be square. As illustrated in Figure~\ref{fig:intro_invertible_matrix_theorem}, by dropping this hypothesis, the characterization breaks down into two: surjectivity and injectivity.

\begin{figure}[h]
  \centering
  \input{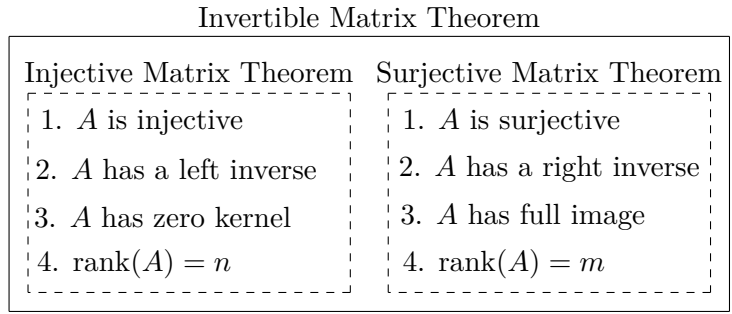}
  \caption{For a general $n \by m$ matrix $A$, surjectivity and injectivity are separate properties.
  However, when $A$ is square, the two characterizations merge into one.}
    \label{fig:intro_invertible_matrix_theorem}
\end{figure}

The Invertible Matrix Theorem may be viewed as a gluing theorem, connecting two characterizations into a single one. The key to this connection is knowing how the properties of surjectivity (SUR), injectivity (INJ), and the dimensional properties (i.e. $m \leq n$ and $m \geq n$) relate to each other. Let's call them the \emph{fundamental properties} of a matrix. A good way to illustrate their relation is through a graph, shown in Figure~\ref{fig:intro_poset_of_4_fundamental_properties}.

\begin{figure}[h]
  \centering
  \input{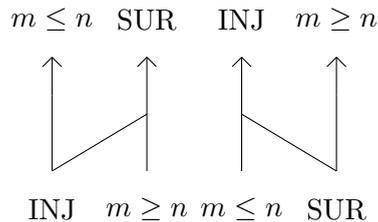}
  \caption{Graph of the $4$ fundamental properties of a $n \by m$ matrix.
  To read this graph, start at a node $Y$ at the top and let $X_1, X_2,\ldots, X_k$ be all the nodes at the bottom reachable through the paths starting at $Y$.
  Then, $X_1,\ldots, X_k \implies Y$.
  For example, if we start at the node SUR at the top and follow the edges downwards, we reach the nodes INJ and $(m \leq n)$ at the bottom.
  }
  \label{fig:intro_poset_of_4_fundamental_properties}
\end{figure}
The graph shows that an injective matrix with more rows than columns must also be surjective. Consequently, for square matrices injectivity and surjectivity are equivalent notions, thus characterizing a single concept: invertibility. Under this perspective, the Invertible Matrix Theorem is a consequence of two results: \emph{a characterization of the fundamental properties} and \emph{a graph showing how they relate}.

Everything so far concerns theorems about a single matrix.
Linear algebra, however, also has various theorems about a more general class of objects, namely the solutions of the linear system $\{(x,y) \mid Ax = By\}$
(Notice that by setting $B = I$, this becomes the graph for the function $A$. Thus, this generalizes the ``single matrix'' case).
In this context, there are natural generalizations of the aforementioned fundamental properties SUR, INJ, $m \leq n$ and $m \geq n$, and also two new ones called determinism (DET) and totality (TOT).
The main goal of this paper is to generalize the Invertible Matrix Theorem to this setting.
More concretely, Theorems~\ref{thm:cospan_dict} and~\ref{thm:poset_of_fundamental_properties} provide characterizations for each fundamental property and a gluing theorem showing how they relate.
This yields a more general version of the graph in Figure~\ref{fig:intro_poset_of_4_fundamental_properties}, shown in Figure~\ref{fig:intro_poset_of_6_fundamental_properties}.

\begin{figure}[h]
  \centering
  \input{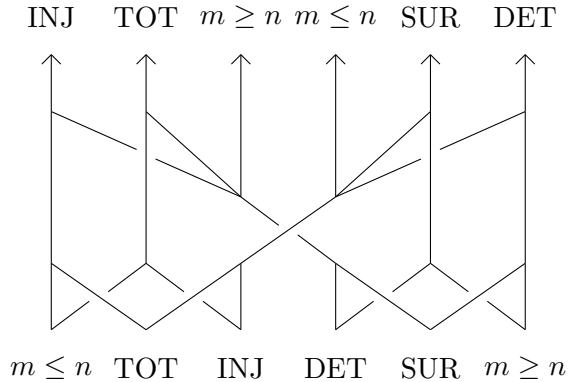}
  \caption{Graph of the $6$ fundamental properties of a pair of $k \by m$ and $k \by n$ matrices.
  To read this graph, use the same rules as Figure~\ref{fig:intro_poset_of_4_fundamental_properties}.}
  \label{fig:intro_poset_of_6_fundamental_properties}
\end{figure}
We also show this is a generalization: by taking $B = I$ the identity matrix, the characterizations reduce to the usual Injective and Surjective Matrix Theorems.
Furthermore, assuming $m = n$ recovers the original Invertible Matrix Theorem.
To better formulate that generalization, the language of \emph{linear relations} proves useful.
Thus, in what follows, we briefly introduce the topic.

\subsection{Linear Relations}
\label{sec:pigeonhole}

One of the simplest ways to motivate the use of relations in linear algebra involves the pigeonhole principle and drawing dots. Given two finite sets $A$ and $B$, we call any subset $R \subseteq A \times B$ a \emph{relation} between $A$ and $B$.
By representing these sets as dots on a sheet of paper, a relation $R$ becomes a graph connecting those dots (we draw a line connecting $a \in A$ to $b \in B$ if and only if $(a, b) \in R$), as seen in Figure~\ref{fig:fin-rel}.

\begin{figure}[h]
    \centering
    \input{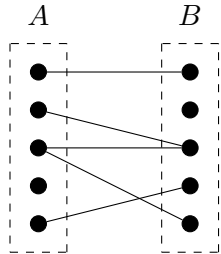}
    \caption{A relation between two finite sets $A$ and $B$ represented as a bipartite graph.}
    \label{fig:fin-rel}
\end{figure}

\noindent It is always possible to represent a function $f : A \to B$ by its graph $\text{Gr}(f) \coloneq \{(a,b) \mid f(a) = b\} \subseteq A \times B$, making functions a special case of relations.
The pigeonhole principle states that if $f$ is injective, then $\# A \leq \# B$, where $\#A$ stands for its cardinality.
Figure~\ref{tab:pigeonhole_2} illustrates the graph of an injective function.
Notice how the pigeonhole principle forces
the set on the right to be at the as large as the set on the left.

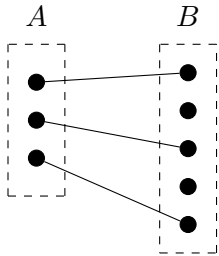
\begin{figure}[h]
    \centering
    \begin{tabular}{l l}
    \input{tikz/figures/introduction_figures/pigeonhole_1.tikz} & \makecell{Function: Any dot on the left connects to a unique dot on the right. \\
    INJ: Any dot on the right connects to at most one dot on the left.}
    \end{tabular}
    \caption{The graph of an injective function.}
    \label{tab:pigeonhole_2}
\end{figure}

Since every dot on the left connects to a different unique dot on the right, there must be at least as many right-dots as left-dots. The not-so-immediate observation is that we do not need uniqueness in the function property for the pigeonhole principle to hold. More concretely, we only need that every dot on the left connects to \emph{at least} one dot on the right.

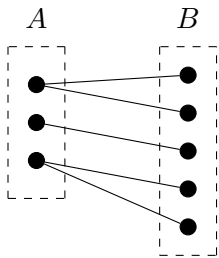
\begin{figure}[h]
    \centering
    \begin{tabular}{l l}
    \input{tikz/figures/introduction_figures/pigeonhole_2.tikz} & \makecell{TOT: Any dot on the left connects to at least one dot on the right. \\
    INJ: Any dot on the right connects to at most one dot on the left.}
    \end{tabular}
    \caption{The pigeonhole principle still works for the graph of a relation satisfying TOT+INJ.}
    \label{tab:pigeonhole_3}
\end{figure}

As seen in Figure~\ref{tab:pigeonhole_3}, the pigeonhole principle holds for some relations that are not functions (i.e. do not satisfy the function property). This tells us that the pigeonhole principle only requires relations and that its usual formulation about injective functions contains an unnecessary hypothesis.

These observations about functions of finite sets generalize to linear transformations. A \emph{linear relation} between vector spaces $V$ and $W$ is a subspace $R \subseteq V \times W$ of their Cartesian product.
By replacing the cardinality $\# X$ with the dimension $\text{dim}(X)$, the pigeonhole principle works the same, while the previous weakening procedure constructs two linear algebraic results shown in Table~\ref{tab:pigeonhole_4}.

\begin{table}[h]
    \centering
    \bgroup
      \renewcommand{\arraystretch}{1}
      \setlength\tabcolsep{1cm}
      \begin{tabular}{l l}
        \textbf{Pigeonhole for functions} & \textbf{Pigeonhole for linear transformations}  \\[1mm]
          If $f : A \to B$ is INJ, & If $f : V \to W$ is INJ, \\
          then $\# A \leq \# B$ & then $\text{dim}(V) \leq \text{dim}(W)$ \\[5mm]
          \textbf{Pigeonhole for relations} & \textbf{Pigeonhole for linear relations} \\[1mm]
          If $R \subseteq A \times B$ is TOT+INJ, & If $R \subseteq V \times W$ is TOT+INJ, \\
          then $\# A \leq \# B$ & then $\text{dim}(V) \leq \text{dim}(W)$
      \end{tabular}
    \egroup
    \caption{The usual pigeonhole principle can be weakened to relations and generalized to the context of linear algebra.}
    \label{tab:pigeonhole_4}
\end{table}

The usual pigeonhole principle of linear transformations is a standard result in linear algebra: injective matrices are tall.
Given how well-understood linear algebra is, the modified, relational pigeonhole principle is expected to also be equivalent to a well-known result.
Perhaps surprisingly, it is equivalent to the Exchange Lemma~\cite[Page 35; 2.22]{axler2024linear}, a fundamental result in linear algebra implying that all bases for a given vector space must have the same size. The precise translation from one to the other are explained later in Section~\ref{sec:aplications_poset_of_fundamental_properties}).
However, it already serves as a good example of how the language of relations produces unexpected connections between seemingly unrelated results. This could not be more appropriate for our goals: as discussed earlier, the Invertible Matrix Theorem is essentially a \emph{connecting} (or gluing) theorem for different fundamental properties.

To close this section, we properly define linear relations.
Also, from now on, we always work over a fixed field $\K$
and, since finite-dimensional vector spaces are completely determined
(up to isomorphism) by their dimension,
we restrict ourselves to subspaces $R \subseteq \K^m \times \K^n$.

\begin{definition}
\label{def:relation}
A \emph{linear relation} $R$ between finite-dimensional vector spaces $V, W$ is a subspace $R \subseteq V \times W$.
\end{definition}

\begin{definition}
\label{def:relation_operations}
The main operations on linear relations are the following.
\begin{enumerate}
\item[(i)] Given $R \subseteq \K^{m} \times \K^k$, $S \subseteq \K^k \times \K^n$, define their \emph{relational composition} $R \,;\, S \subseteq \K^m \times \K^n$ as
$$R \,;\, S \coloneq \{(x, z) \mid \exists y ((x, y) \in R, (y, z) \in S)\}.$$
\item[(ii)] Given $R \subseteq \K^m \times \K^n$, $S \subseteq \K^s \times \K^r$, define their \emph{Cartesian product}
$R \times S \subseteq \K^{m+s} \times \K^{n+r}$ as
$$R \times S \coloneq \{((x, w), (y, z)) \mid (x, y) \in R, (w, z) \in S\}.$$
\item[(iii)] Given $R \subseteq \K^m \times \K^n$, define its \emph{opposite} $R^{op} \subseteq \K^n \times \K^m$ as
$$R^{op} \coloneq \{(x,y) \mid (y, x) \in R\}.$$
\end{enumerate}
\end{definition}

\begin{definition}
\label{def:relation_poset}
Linear relations have a natural poset structure $\subseteq$ induced by the usual set-theoretic inclusion. For linear relations $R, S \in \K^n \times \K^m$,

\[
  R \subseteq S
  \iff
  \forall x \in \K^n,\, y \in \K^m,\; (x,y) \in R \implies (x, y) \in S.\]

\end{definition}

\noindent Linear relations are generalizations of linear maps and linear subspaces, in the sense that any subspace $V \subseteq \K^n$ can be seen as a $0$-by-$n$ linear relation $\{(0, x) \mid x \in v\} \subseteq \K^{0} \times \K^n$, where $\K^0 = \{0\}$ is the zero-dimensional vector-space, and any linear map $f : \K^m \to \K^n$ can be seen as a $m$-by-$n$ linear relation via its graph $\text{Gr}(f) \coloneq \{(x, f(x)) \mid x \in \K^m\}$. When restricted to linear maps, the operations (i) and (ii) above become the usual composition and Cartesian product of maps. Moreover, the composition $V \,;\, f$ is the subspace $f(V)$.

Relational composition and Cartesian product are associative operations, so they satisfy the laws $(R \,;\, S) \,;\, H = R \,;\, (S \,;\, H)$ and $(R \times S) \times H = R \times (S \times H)$. There is also a law describing how these two operations interact:
\begin{equation}\label{eq:2d_associativity_in_symbolic_syntax}
    (A \,;\, B) \times (C \,;\, D) = (A \times C) \,;\, (B \times D).
\end{equation}
As hinted by the discussion about the Pigeonhole Principle, the notions of surjectiveness and injectiveness extend to linear relations, together with two dual properties named \emph{totality} and \emph{determinism}.
\begin{definition}\label{def:fundamental_properties_classical} A linear relation $R \subseteq V \times W$ is \\
\begin{tabular}{r r l}
  (i) & injective if & $\forall y \in W \text{ there is at most one $x \in V$ such that } (x,y) \in R;$ \\
  (ii) & surjective if & $\forall y \in W \text{ there is at least one $x \in V$ such that } (x,y) \in R;$ \\
  (iii) & deterministic if & $\forall x \in V \text{ there is at most one $y \in W$ such that } (x,y) \in R;$ \\
  (iv) & total if & $\forall x \in V \text{ there is at least one $y \in W$ such that } (x,y) \in R.$
\end{tabular}
\end{definition}

\subsection{Graphical Notation}
\label{sec:graphical_notation}

In this paper, linear relations are represented using string diagrams.
That is, instead of the usual symbolic syntax,
relations are written as two-dimensional diagrams and composed similarly to logical gates.
This is known as \emph{graphical notation} and is widely used among category theorists.
According to \cite{hinze2023introducing},
diagrams enable equational reasoning without loss of type information.
The use of diagrams in linear algebra as seen here was developed by Zanasi in his thesis~\cite{ZanasiThesis},
leading to what is now called Graphical Linear Algebra (GLA)~\cite{bonchi2014categorical,Bonchi2015,bonchi2017refinement,Bonchi2019a, PAIXAO2022}.

In what follows we give a brief introduction
to enable the reader to read proofs written with diagrams.
For a more thorough account of Graphical Linear Algebra,
we recommend the blog series by Pawe{\l} Soboci\'{n}ski~\cite{blogpawel}.

A $m$-by-$n$ linear relation can be written as a two-dimensional diagram with $m$ dangling wires on the left and $n$ on the right.
\begin{equation}
R \subset \K^m \times \K^n \mapsto \quad \input{tikz/figures/introduction_figures/n_wired_dots_relation.tikz}.
\end{equation}
For example, a $1$-by-$2$ linear relation $R$ can be written as a diagram \input{tikz/figures/introduction_figures/1_by_2_relation.tikz} with one wire on the left and two wires on the right. Generally, when $R$ is a $m$-by-$n$ linear relation, it is more convenient to write \input{tikz/figures/introduction_figures/m_by_n_relation.tikz}, and when there is no confusion about the type information, we may omit the numbers $m, n$ and write it as \Rel{R}. Lastly, when one of the numbers is zero, say, $m = 0$,  then $R$ is called a linear subspace and we omit the dangling wire on the left, writing \input{tikz/figures/introduction_figures/0_by_n_relation.tikz}.

We also reserve different symbols based on how specialized the linear relation is:
while the box symbol $\Rel{R}$ denotes a general linear relation,
the curvy symbol $\Map{A}$ is used when we know $A$ is a linear map.
Moreover, the blue curvy symbol $\Inv{A}$ is used when $A$ is known to have an inverse,
usually denoted $A^{-1}$.
In the diagrammatic notation, operations~(i) and~(ii) of Definition~\ref{def:relation_operations} become, respectively, the action of connecting diagrams horizontally and stacking them vertically, and the action of flipping the diagram horizontally.
Table~\ref{tab:graphical_notation} summarizes what we have said so far.

\begin{table}[htb]
    \centering
    \input{tikz/figures/introduction_figures/graphical_notation.tikz}
    \caption{Graphical notation.}
    \label{tab:graphical_notation}
\end{table}

Under the graphical syntax, both sides in Equation~\eqref{eq:2d_associativity_in_symbolic_syntax} become the two ways we can put parenthesis in the diagram composition
\begin{equation}
\input{tikz/figures/introduction_figures/2d_associativity.tikz}
\end{equation}
which is very similar to an associativity law. Informally speaking, the difference is that the parentheses move in two dimensions, instead of simply from left-to-right. Due to this law, we will always omit the dashed lines when composing diagrams.

There are specific symbols for some commonly occurring linear relations called identity, twist, zero, sum, copy, and discard. These are all graphs of linear maps. Table~\ref{tab:generators} shows their corresponding diagrammatic symbols and linear maps.

\begin{table}[htb]
  \centering
  \renewcommand{\arraystretch}{2.5}
  \begin{tabular}{r l l l}
    (Identity) & \TypedId{n} & $x \mapsto x$ & $: \K^n \to \K^n$  \\
    (Twist) & \input{tikz/figures/introduction_figures/n_wired_twist.tikz} & $(y \in \K^m, x \in \K^n) \mapsto (y, x)$ & $: \K^{m+n} \to \K^{m+n}$ \\
    (Zero) & \TypedZero{n} & $x \mapsto 0$ & $: \K^0 \to \K^n$ \\
    (Sum) & \input{tikz/figures/introduction_figures/n_wired_sum.tikz} & $(x \in \K^n, y \in \K^n) \mapsto x + y$ & $: \K^{2n} \to \K^n$  \\
    (Discard) & \TypedDiscard{n} & $x \mapsto 0$ & $: \K^n \to \K^0$ \\
    (Copy) & \input{tikz/figures/introduction_figures/n_wired_copy.tikz} & $x \mapsto (x, x)$ & $: \K^n \to \K^{2n}$ \\
  \end{tabular}
  \caption{Special relations together with their names and graphical notation.}
  \label{tab:generators}
\end{table}

In the same way as before, we will often omit the numbers $m,n$ when there is no potential for confusion. The zero diagram \TypedZero{n} can be thought of as the zero subspace $\{0\} \subseteq \K^n$, whereas the opposite of the discard, \TypedCoDiscard{n}, can be thought of as the full space $\K^n \subseteq \K^n$.
Additionally, any linear map $A$ satisfies
\begin{equation}\label{eq:map_equations}
\begin{tabular}{l l}
    $\Image{A} = \; \input{tikz/figures/introduction_figures/image_of_A.tikz}$, & $\Kernel{A} = \; \input{tikz/figures/introduction_figures/kernel_of_A.tikz}$, \\
    $\CoImage{A} = \; \CoDiscard$, & $\CoKernel{A} = \; \Zero$.
\end{tabular}
\end{equation}
From this, we see that a map $A$ is surjective if and only if \Image{A} = \CoDiscard and injective if and only if \Kernel{A} = \Zero. This is still true for linear relations in general. More concretely, all four properties in Definition~\ref{def:fundamental_properties_classical} can be written in terms of the white and black structures, as stated below.
\begin{proposition}[Fundamental Properties]
  \label{def:fundamental_properties_and_invertibility}
  A linear relation \Rel{R} is said to be
\begin{center}
\renewcommand{\arraystretch}{2.0}
\resizebox{\textwidth}{!}{\begin{tabular}{l l l l}
    total (TOT) if &\RelCoImage{R} \!\!= \CoDiscard \!\!\!, & deterministic (DET) if &\RelCoKernel{R} \!\!= \Zero \!\!\!, \\
    surjective (SUR) if &\RelImage{R} \!\!= \CoDiscard \!\!\!, & injective (INJ) if &\RelKernel{R} \!\!= \Zero \!\!\!, \\
    a map if it is &TOT and DET, &  bijective if it is & TOT, DET, INJ, and SUR.
\end{tabular}}
\end{center}
\end{proposition}

\noindent By our previous considerations, this definition coincides with the usual definition of surjectiveness and injectiveness when \Rel{R} is a linear map.
Also, notice that the two inequalities
\begin{equation}
    \RelImage{R} \subseteq \; \CoDiscard, \quad \RelKernel{R} \supseteq \; \Zero
\end{equation}
hold for any relation.
Therefore, throughout the paper, when showing that a relation is TOT, DET, INJ or SUR, we often do not prove the equality required by Definition~\ref{def:fundamental_properties_and_invertibility}, and instead only show the corresponding non-trivial inequality. For example, to show that a linear relation $\Rel{R}$ is total, it is enough to show that $\RelImage{R} \supseteq \; \CoDiscard$.
See Figure~\ref{fig:Inequality} for which inequalities hold for every relation.

As this paper deals with the Invertible Matrix Theorem in a relational context,
we define below what it means for a map to have an inverse.
Notice that we do not assume a bijective map has an inverse,
as this will result from the Invertible Matrix Theorem, proven later in the text.
\begin{definition}
\label{def:inverse}
  Given a map $\Map{A}$, a map $\Map{B}$ is called
  \begin{enumerate}
      \item[(i)] a left-inverse\footnote{It is worth noting that the left-inverse $\Map{B}$ appears at the \emph{right} of $\Map{A}$, and the right-inverse appears at the \emph{left}. This is because the relational composition $A ; B$ has its arguments flipped with respect to the standard composition (i.e. the composition of matrices $BA$ is written with diagrams as $\Comp{A/Map, B/Map}$). We maintain this confusing terminology in order to stay consistent with the standard literature.} if $\Comp{A/Map, B/Map} = \Id$,
      \item[(ii)] a right-inverse if $\Comp{B/Map, A/Map} = \Id$,
      \item[(iii)] an inverse if it is a right-inverse and a left-inverse.
  \end{enumerate}
\end{definition}
\noindent Also, notice that an inverse is unique
and that if a map has both a left and a right-inverse, they must be equal
and the map has an inverse.

It is important to note that the poset structure $\subseteq$
interacts well with the two composition operations~\cite{PAIXAO2022}.
More concretely, if $\Rel{X} \subseteq \Rel{Y}$, then for any relations $\Rel{A}, \Rel{B}$ we have
\begin{align*}
  \input{tikz/figures/introduction_figures/poset_vs_composition_1.tikz}
  &,&
  \input{tikz/figures/introduction_figures/poset_vs_composition_2.tikz}.
\end{align*}
We will often use this to produce inequalities between complicated diagrams. For example, we know $\Zero \subseteq \CoDiscard$, therefore, the following inequality must hold, regardless of what $A$ and $B$ are:
\[\input{tikz/figures/introduction_figures/poset_vs_composition_3.tikz}.\]

All proofs in this paper consist of sequences of diagram manipulations.
Appendix~\ref{sec:gla_thms} contains all necessary rules in their diagrammatic form.
The proof style follows a calculational format~\cite{Dijkstra1989PredicateCA},
with a ``proof hint'' in square brackets to the right of a diagram or statement
explaining which result was used  to conclude it from the previous one.
To exemplify, here is the proof that any map $\Map{A}$ with a left-inverse $\Map{B}$
is injective:
\begin{hcalculation}[\subseteq]{\Kernel{A}}
  \hstep{\CompTwo{White, CoMap, B, CoMap, A, .}}{Fig.\ref{fig:Minimum_and_maximum}: $\Zero \subseteq \Kernel{}$}
  \hstep[=]{\Zero}{Hyp: $\Comp{A/Map, B/Map} = \Id$}
\end{hcalculation}
On the right, the manipulation rules in Figure~\ref{fig:Minimum_and_maximum}
and the theorem's hypothesis
are noted as the steps for going from the first diagram to the last one.

The propositions in Figure~\ref{fig:Symmetry},
reproduced below,
\[
  \begin{aligned}
    \Rel{R} & \subseteq \Rel{S} &\iff \CoRel{R} &\subseteq \CoRel{S}, \\
    \RelGray{R} & \subseteq \Rel{S} &\iff \Rel{R} &\supseteq \RelGray{S}.
  \end{aligned}
\]
imply that any inequality between relations gives rise to other two inequalities
by either mirroring all relations to their opposites or swapping their colors.
This way, any theorem in diagrammatic form actually summarizes 4 theorems.
As a corollary, the opposites and color-swapped versions of an \emph{equality} are also equalities.
This will be used thoroughly in this paper without further comments.

Finally, we note that the advantages of diagram manipulation were the main reason for choosing graphical notation as a tool in this work.
However, it is important to emphasize that the results obtained do not depend exclusively on this notation;
all proofs can be equivalently written in the usual syntax.

\subsection{Main Contributions}

This paper main contributions are the characterization of fundamental properties
in Theorem~\ref{thm:cospan_dict} generalizing the Invertible Matrix Theorem for linear relations
and the graph of fundamental properties implications in Theorem~\ref{thm:poset_of_fundamental_properties}.
To tackle those results, in Theorem~\ref{thm:rcd_for_cospan},
we develop a decomposition for linear relations in cospan form
\[ \Rel{R} = \CoSpan{A}{B} = \GRR,\]
which acts as a generalization of the Canonical Decomposition for matrices.

Furthermore, this decomposition induces a \emph{generalized matrix decomposition}
simultaneously decomposing the cospan's factors as
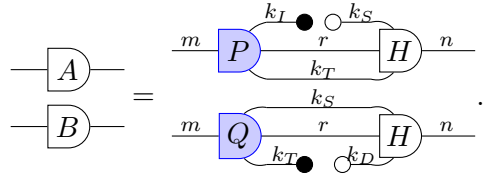
\[\input{tikz/figures/rcd_for_pairs_of_matrices/stat.tikz}.\]
The linking matrix $\Map{H}$ is constructed in Theorems~\ref{thm:abstract} and~\ref{thm:rcd_for_pairs_of_matrices},
and, as shown in Theorem~\ref{thm:subspaces},
can be used to calculate all subspaces related to the images $\im(A)$ and $\im(B)$.
It is also worth noticing that this decomposition is analogous to the Generalized Singular Value Decomposition (GSVD)~\cite{gsvd1976,gsvd1981},
but defined without assumptions on the underlying scalar field.

\section{Fundamental Properties}\label{sec:fundamental_properties}

We start our discussion by dissecting the proof of
the original Invertible Matrix Theorem (IMT),
as a motivation for the decisions taken later in the text regarding its relational version.
As mentioned in Section~\ref{sec:introduction},
IMT is a two-part theorem, starting with characterizing surjectivity and injectivity
i.e., the fundamental properties of a linear map.
For an $m$-by-$n$ linear map $A$,
such a characterization can be split into two blocks,
as illustrated in Figure~\ref{fig:imt-graph},
\begin{enumerate}
  \item Statements regarding the existence of a right or left-inverse $B$ to $A$ or, equivalently, solutions to linear systems.
  \item Statements concerning only the map $A$ and its properties;
\end{enumerate}

\begin{figure}[h]
    \centering
    \input{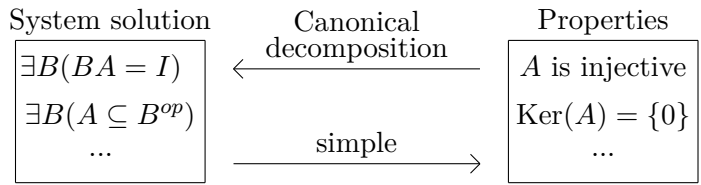}
    \caption{Statements in the same block can easily be shown to be equivalent. However, going from a right-statement to a left-statement requires a matrix decomposition or an equivalent algorithm such as Gaussian elimination to construct the linear map $B$.}
    \label{fig:imt-graph}
\end{figure}

Statements in the same block are straightforwardly shown to be equivalent.
Also, the existence of an inverse $B$ directly implies the properties of interest,
such as injectivity or trivial kernel.
The hard part of proving the IMT is showing that those properties are enough to to assure the existence of an inverse.
This process involves finding an algorithm (e.g. Gaussian elimination)
to construct said linear map $B$.
Instead of reproducing the algorithm inside the proof,
this is usually appears as a rank-revealing matrix decomposition
\[
  \label{eq:functional_canonical_decomposition}
  A = P^{-1}\begin{bmatrix} I & 0 \\ 0 & 0 \\\end{bmatrix} Q^{-1}.
\]

Once the characterizations are proven, the decomposition in Equation~\eqref{eq:functional_canonical_decomposition}
can be used to derive a graph connecting the properties SUR and INJ with the dimensionality properties $m \leq n$ and $m \geq n$ shown in Figure~\ref{fig:intro_poset_of_4_fundamental_properties}.
This gives a general guideline of how to prove the IMT:
\begin{enumerate}
    \item[(i)] For each of the properties SUR, INJ, prove the equivalencies in both blocks separately;
    \item[(ii)] Connect the two blocks via a canonical decomposition;
    \item[(iii)] Use the decomposition once more to build a graph connecting the fundamental properties SUR and INJ to the dimensionality properties $m \leq n$ and $m \geq n$.
\end{enumerate}

\subsection{Preliminary Characterizations}

We now follow these general guidelines to generalize the IMT to the context of linear relations.
First, we state a classical result,
asserting that any linear relation $R$ can be realized as the solution set $\{(x, y) \mid Ax = By\}$
for some pair of linear maps $A, B$.
This is known as the \emph{cospan form} of $R$.

\begin{proposition}\label{prop:cospan}
For every linear relation $\TypedRel{R}{m}{n}\!$, there exist $\TypedMap{A}{m}{k}\!\!$, $\TypedMap{B}{n}{k}\!\!$ linear maps such that
\[\Rel{R} = \CoSpan{A}{B}.\]
\end{proposition}
\begin{proof}
   See Lemma \ref{lem:relation_in_cospan_form_complete}.
\end{proof}

Due to this result, we express the following theorems in terms of the maps $\Map{A}$ and $\Map{B}$ instead of the linear relation $\Rel{R}$.
As in Definition~\ref{def:fundamental_properties_and_invertibility}, there are two additional fundamental properties called determinism and totality.
The first step is to prove a characterization theorem for each fundamental property
without involving invertibility or the existence of solutions.

\begin{proposition}
  \label{prop:r_block}
  Given linear maps $\TypedMap{A}{m}{k}\!\!$, $\TypedMap{B}{n}{k}\!\!$, each column of the following table consists of a series of equivalent statements.
\[\resizebox{\textwidth}{!}{\input{tikz/figures/thm_cospan_dict/R_block.tikz}}\]
\end{proposition}
\begin{proof}
We first show the TOT case, then the DET case.
The other two follow similarly.
The strategy is Row 1 $\Rightarrow$ Row 2 $\Rightarrow$ Row 3 $\Rightarrow$ Row 1.

(Row 1 $\implies$ Row 2)
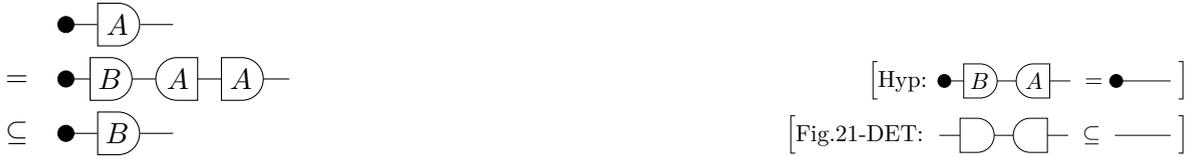
\begin{hcalculation}[\subseteq]{\Image{A}}
\hstep[=]{\CompThree{Black, Map, B, CoMap, A, Map, A, .}}{Hyp: $\CompTwo{Black, Map, B, CoMap, A, .} = \CoDiscard$}
\hstep{\Image{B}}{Fig.\ref{fig:typerelations}-DET: $\CoSpan{\;}{\;} \subseteq \Id$}
\end{hcalculation}

(Row 2 $\implies$ Row 3)
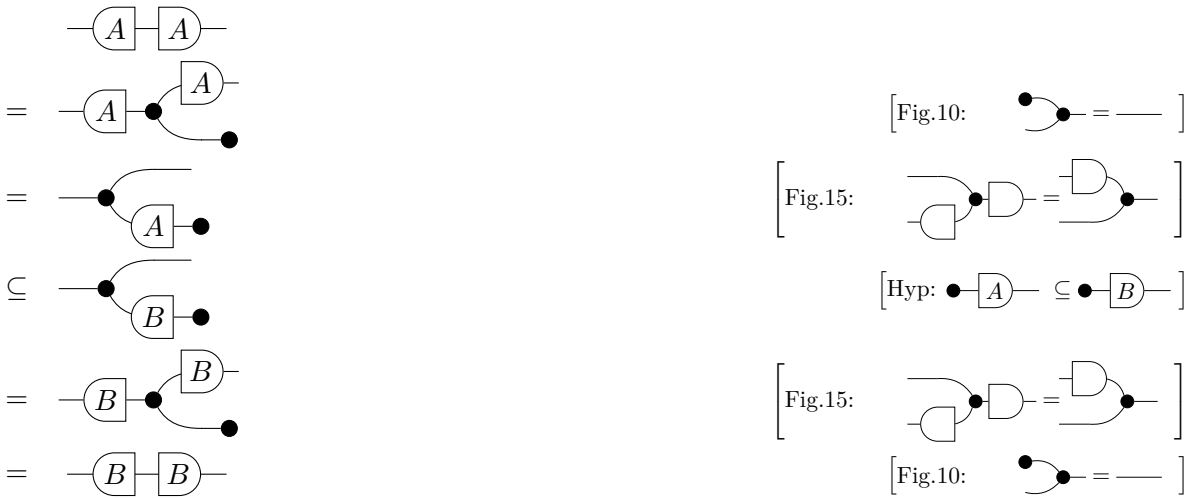
\begin{hcalculation}[=]{\Span{A}{A}}
\hstep{\input{tikz/figures/thm_cospan_dict/row_2_implies_row_3_step_2.tikz}}{Fig.\ref{fig:Commutative_Comonoid}: \input{tikz/figures/theorems/unitblack.tikz}}
\hstep{\input{tikz/figures/thm_cospan_dict/row_2_implies_row_3_step_3.tikz}}{Fig.\ref{fig:Flip}: \input{tikz/figures/theorems/flipblack.tikz}}
\hstep[\subseteq]{\input{tikz/figures/thm_cospan_dict/row_2_implies_row_3_step_4.tikz}}{Hyp: $\Image{A} \subseteq \Image{B}$}
\hstep{\input{tikz/figures/thm_cospan_dict/row_2_implies_row_3_step_5.tikz}}{Fig.\ref{fig:Flip}: \input{tikz/figures/theorems/flipblack.tikz}}
\hstep{\Span{B}{B}}{Fig.\ref{fig:Commutative_Comonoid}: \input{tikz/figures/theorems/unitblack.tikz}}
\end{hcalculation}

(Row 3 $\implies$ Row 1)
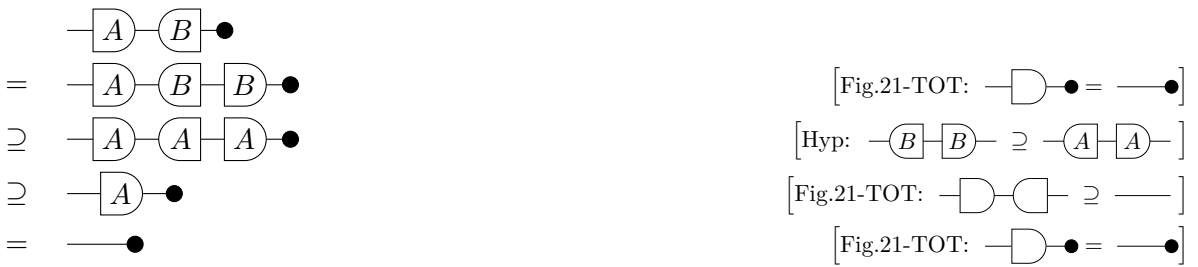
\begin{hcalculation}[=]{\CompTwo{., Map, A, CoMap, B, Black}}
\hstep{\CompThree{., Map, A, CoMap, B, Map, B, Black}}{Fig.\ref{fig:typerelations}-TOT: $\Diagram{Map}{}{.}{Black} = \Discard$}
\hstep[\supseteq]{\CompThree{., Map, A, CoMap, A, Map, A, Black}}{Hyp: $\Span{B}{B} \supseteq \Span{A}{A}$}
\hstep[\supseteq]{\Diagram{Map}{A}{.}{Black}}{Fig.\ref{fig:typerelations}-TOT: $\CoSpan{\;}{\;} \supseteq \Id$}
\hstep{\Discard}{Fig.\ref{fig:typerelations}-TOT: $\Diagram{Map}{}{.}{Black} = \Discard$}
\end{hcalculation}

\noindent Now we show the DET column. The proof is analogous, but we include it for the sake of clarity.

(Row 1 $\implies$ Row 2)
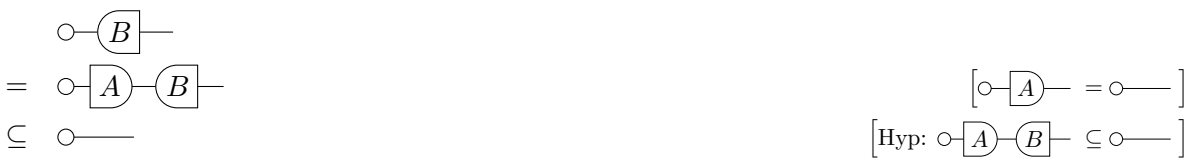
\begin{hcalculation}[\subseteq]{\Kernel{B}}
\hstep[=]{\CompTwo{White, Map, A, CoMap, B, .}}{$\CoKernel{A} = \Zero$}
\hstep[\subseteq]{\Zero}{Hyp: $\CompTwo{White, Map, A, CoMap, B, .} \subseteq \Zero$}
\end{hcalculation}

(Row 2 $\implies$ Row 3)
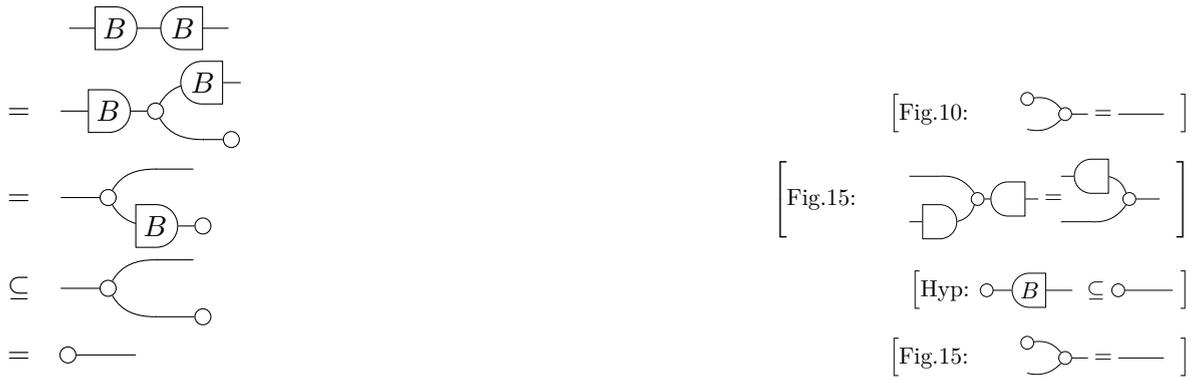
\begin{hcalculation}[=]{\CoSpan{B}{B}}
\hstep{\input{tikz/figures/thm_cospan_dict/row_2_implies_row_3_step_2_DET.tikz}}{Fig.\ref{fig:Commutative_Comonoid}: \input{tikz/figures/theorems/unit.tikz}}
\hstep{\input{tikz/figures/thm_cospan_dict/row_2_implies_row_3_step_3_DET.tikz}}{Fig.\ref{fig:Flip}: \input{tikz/figures/theorems/flip.tikz}}
\hstep[\subseteq]{\input{tikz/figures/thm_cospan_dict/row_2_implies_row_3_step_4_DET.tikz}}{Hyp: $\Kernel{B} \subseteq \Zero$}
\hstep{\Zero}{Fig.\ref{fig:Flip}: \input{tikz/figures/theorems/unit.tikz}}
\end{hcalculation}

(Row 3 $\implies$ Row 1)
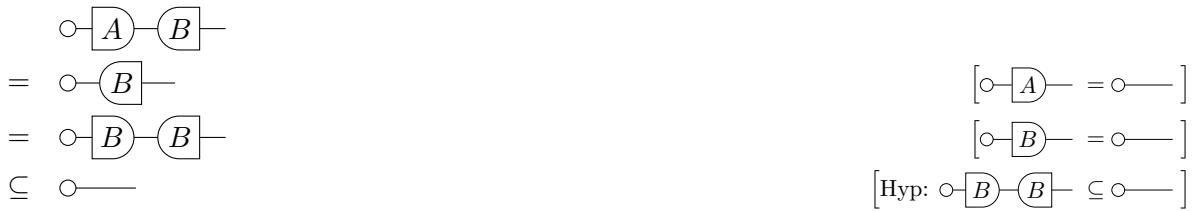
\begin{hcalculation}[=]{\CompTwo{White, Map, A, CoMap, B, .}}
\hstep{\Kernel{B}}{$\CoKernel{A} = \Zero$}
\hstep{\CompTwo{White, Map, B, CoMap, B, .}}{$\CoKernel{B} = \Zero$}
\hstep[\subseteq]{\Zero}{Hyp: $\CompTwo{White, Map, B, CoMap, B, .} \subseteq \Zero$}
\end{hcalculation}
\end{proof}

Now we move our focus to an equivalence for existence theorem
about the constituent parts of a cospan.
These results are related to invertibility, as we will shortly see.

\begin{proposition}
  \label{prop:s_block}
  For maps $\TypedMap{A}{m}{k}\!\!$, $\TypedMap{B}{n}{k}\!\!$,
\begin{enumerate}
  \item[(i)] $\exists \Map{S_1}, \Map{S_1} \subseteq \CoSpan{A}{B}
      \iff
      \exists \Map{S_1}, \Comp{S_1/Map, B/Map} = \Map{A} $,
  \item[(ii)] $\exists \Map{S_2}, \CoMap{S_2} \subseteq \CoSpan{A}{B}
      \iff
      \exists \Map{S_2}, \Comp{S_2/Map, A/Map} = \Map{B} $.
\end{enumerate}
\end{proposition}

\begin{proof}
  We prove item (i). The other is analogous.
  \begin{enumerate}
    \item[$(\Rightarrow)$]
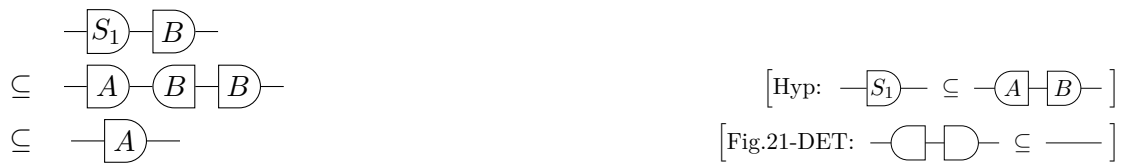
    \begin{hcalculation}[=]{\Comp{S_1/Map, B/Map}}
      \hstep[\subseteq]{\Comp{A/Map, B/CoMap, B/Map}}{Hyp: $\Map{S_1} \subseteq \Span{A}{B}$}
      \hstep[\subseteq]{\Map{A}}{Fig.\ref{fig:typerelations}-DET: $\Span{\;}{\;} \subseteq \Id$}
    \end{hcalculation}
    The equality follows by Lemma~\ref{lem:pair_of_matrices_1}.
    \item[$(\Leftarrow)$]
    \begin{hcalculation}[=]{\Map{S_1}}
      \hstep[\subseteq]{\Comp{S_1/Map, B/Map, B/CoMap}}{Fig.\ref{fig:typerelations}-TOT: $\Id \subseteq \CoSpan{\;}{\;}$}
      \hstep{\Comp{A/Map, B/CoMap}}{Hyp: $\Comp{S_1/Map, B/Map} = \Map{A}$}
    \end{hcalculation}
  \end{enumerate}
\end{proof}

\begin{remark}
Notice that although the statements in the previous theorem
are in terms of existential propositions,
the proofs guarantee that the quantified maps
on both sides of the ``if and only if'' are the same.
\end{remark}

Recall Definition~\ref{def:inverse} about left and right-inverses.
From Proposition~\ref{prop:s_block}, we can derive some properties about a map and its inverse.
\begin{proposition}
\label{thm:prop_linverse}
If a linear map $\Map{A}$ has a left-inverse $\Map{B}$, the following hold for all linear relations $\Rel{R}, \Rel{S}$.
\begin{enumerate}
    \item[(i)] \Map{A} is INJ,
    \item[(ii)] $\Map{A} \subseteq \CoMap{B}$,
    \item[(iii)] $\Comp{R/Rel, A/Map} = \Rel{S} \implies \Rel{R} = \Comp{S/Rel, B/Map}$.
    \item[(iv)] $\Rel{R} = \Rel{S} \iff \Comp{R/Rel, A/Map} = \Comp{S/Rel, A/Map}$.
\end{enumerate}
\end{proposition}
\begin{proof} \quad

(i)
\begin{hcalculation}[\subseteq]{\Kernel{A}}
  \hstep{\CompTwo{White, CoMap, B, CoMap, A, .}}{Fig.\ref{fig:Minimum_and_maximum}: $\Zero \subseteq \Kernel{}$}
  \hstep[=]{\Zero}{Hyp: $\Comp{A/Map, B/Map} = \Id$}
\end{hcalculation}

(ii) Proposition~\ref{prop:s_block} by taking $\Map{S} \coloneq \Map{B}$.

(iii) Here,
\begin{hcalculation}{\Rel{R}}
\hstep{\Comp{R/Rel, A/Map, B/Map}}{Hyp: $\Comp{A/Map, B/Map} = \Id$}
\hstep{\Comp{S/Rel, B/Map}}{Hyp: $\Comp{R/Rel, A/Map} = \Rel{S}$}
\end{hcalculation}

(iv) The $(\Rightarrow)$ direction is straightforward,
\begin{hcalculation}[\implies]{\Rel{R} = \Rel{S}}
\hstep{\Comp{R/Rel, A/Map} = \Comp{S/Rel, A/Map}}{Compose with $\Map{A}$}
\end{hcalculation}

For $(\Leftarrow)$,
\begin{hcalculation}[\implies]{\Comp{R/Rel, A/Map} = \Comp{S/Rel, A/Map}}
\hstep{\Comp{R/Rel, A/Map, B/Map} = \Comp{S/Rel, A/Map, B/Map}}{Compose with $\Map{B}$}
\hstep{\Rel{R} = \Rel{S}}{Hyp: $\Comp{A/Map, B/Map} = \Id$}
\end{hcalculation}

The $(\Rightarrow)$ direction is relation composition.
\end{proof}

Analogously, right-inverses satisfy a dual version of this result.
\begin{proposition}
\label{thm:prop_rinverse}
If a linear map $\Map{A}$ has a right-inverse $\Map{B}$, the following hold for all linear relations $\Rel{R}, \Rel{S}$.
\begin{enumerate}
    \item[(i)] \Map{A} is SUR,
    \item[(ii)] $\Map{A} \supseteq \CoMap{B}$,
    \item[(iii)] $\Comp{A/Map, R/Rel} = \Rel{S} \implies \Rel{R} = \Comp{B/Map, S/Rel}$.
    \item[(iv)] $\Rel{R} = \Rel{S} \iff \Comp{A/Map, R/Rel} = \Comp{A/Map, S/Rel}$.
\end{enumerate}
\end{proposition}

By joining these two results, we further establish properties for maps with inverses.
\begin{proposition}
\label{thm:prop_inverse}
If a linear map $\Map{A}$ has an inverse $\Map{B}$, the following hold for all linear relations $\Rel{R}, \Rel{S}$.
\begin{enumerate}
    \item[(i)] \Map{A} is bijective,
    \item[(ii)] $\Map{A} = \CoMap{B}$,
    \item[(iii)] $\Comp{R/Rel, A/Map} = \Rel{S} \iff \Rel{R} = \Comp{S/Rel, B/Map}$.
    \item[(iv)] $\input{tikz/figures/thm_cospan_dict/inv_cancellation_law.tikz}$
\end{enumerate}
\end{proposition}

The next step is to connect Propositions~\ref{prop:r_block} and~\ref{prop:s_block}
through a canonical decomposition.
The usual canonical decomposition of linear maps is not sufficient in this case.
Thus, the need arises for a relational version of Equation~\eqref{eq:functional_canonical_decomposition}.
The goal of the next section is to construct such a decomposition.

\subsection{Decomposing Linear Relations}
\label{sec:decomposition_relations}

Gaussian elimination is among the most widely used tools in linear algebra.
Although it is most commonly associated with the LU matrix decomposition, Gaussian elimination also appears in other matrix factorization methods.
Particularly useful for linear relations is the \emph{Canonical Decomposition} from Equation~\eqref{eq:functional_canonical_decomposition},
which factorizes any linear map $\Map{A}$ in terms of matrices with inverses, linked by $r$ simple wires, where $r$ is the rank of $\Map{A}$.

\begin{theorem}[Canonical Decomposition]\label{thm:gaussian_elimination}
    For any linear map $\TypedMap{A}{m}{n}$ there are linear maps $\TypedInv{P}{m}{m}, \TypedInv{Q}{n}{n}$ with inverses $\Inv{P^{-1}}, \Inv{Q^{-1}}$ and a number $r \in \mathbb{N}$ such that
    \[\TypedMap{A}{m}{n} = \GaussianElim.\]
\end{theorem}
\begin{proof}
When transforming a matrix $\Map{A}$ into reduced row-echelon form (RREF) we obtain
\begin{equation}\label{eq:RREF}
    \Map{A} = \input{tikz/figures/rcd_for_cospan/cospan/canonicaldecomposition/step1.tikz}
\end{equation}
where $\Inv{Q^{-1}}$ is an elementary matrix that has inverse and $\Inv{P_{er}}$ is a permutation matrix.
By rearranging the diagrams, we get

\begin{hcalculation}[=]{\Map{A}}
\hstep{\input{tikz/figures/rcd_for_cospan/cospan/canonicaldecomposition/step1.tikz}}{Eq.\eqref{eq:RREF}: \Map{A} = \input{tikz/figures/rcd_for_cospan/cospan/canonicaldecomposition/step1.tikz}}
\hstep{\input{tikz/figures/rcd_for_cospan/cospan/canonicaldecomposition/step2.tikz}}{Fig.\ref{fig:Commutative_Comonoid}: \input{tikz/figures/theorems/unit.tikz}}
\hstep{\input{tikz/figures/rcd_for_cospan/cospan/canonicaldecomposition/step3.tikz}}{\input{tikz/figures/rcd_for_cospan/cospan/canonicaldecomposition/defP.tikz}}
\hstep{\input{tikz/figures/rcd_for_cospan/cospan/canonicaldecomposition/step4.tikz}}{Prop.\ref{thm:prop_inverse}: $\Inv{P}$ and $\Inv{Q}$ have inverses}
\end{hcalculation}
\end{proof}

The previous theorem applies only to a linear map.
To generalize it to any linear relation,
it is also necessary to reveal the information for totality and determinism.
This gives rise to a decomposition with additional associated numbers besides $r$.

This decomposition is the main tool developed in this paper and,
similarly to the Canonical Decomposition or Gaussian elimination for matrices,
is what allows us to explicitly construct solutions to relational linear equations.
Also, notice that this decomposition was already independently developed in~\cite[Lemma~48]{booth2024completeequationaltheoriesclassical}.

\begin{theorem}[Cospan Decomposition]\label{thm:rcd_for_cospan}
For all linear relations $\Rel{R}$ there are linear maps $\TypedInv{P}{m}{m}, \TypedInv{Q}{n}{n}$ with inverses $\Inv{P^{-1}}$ and $\Inv{Q^{-1}}$ and numbers $k_I, k_S, k_T, k_D, r \in \mathbb{N}$ such that
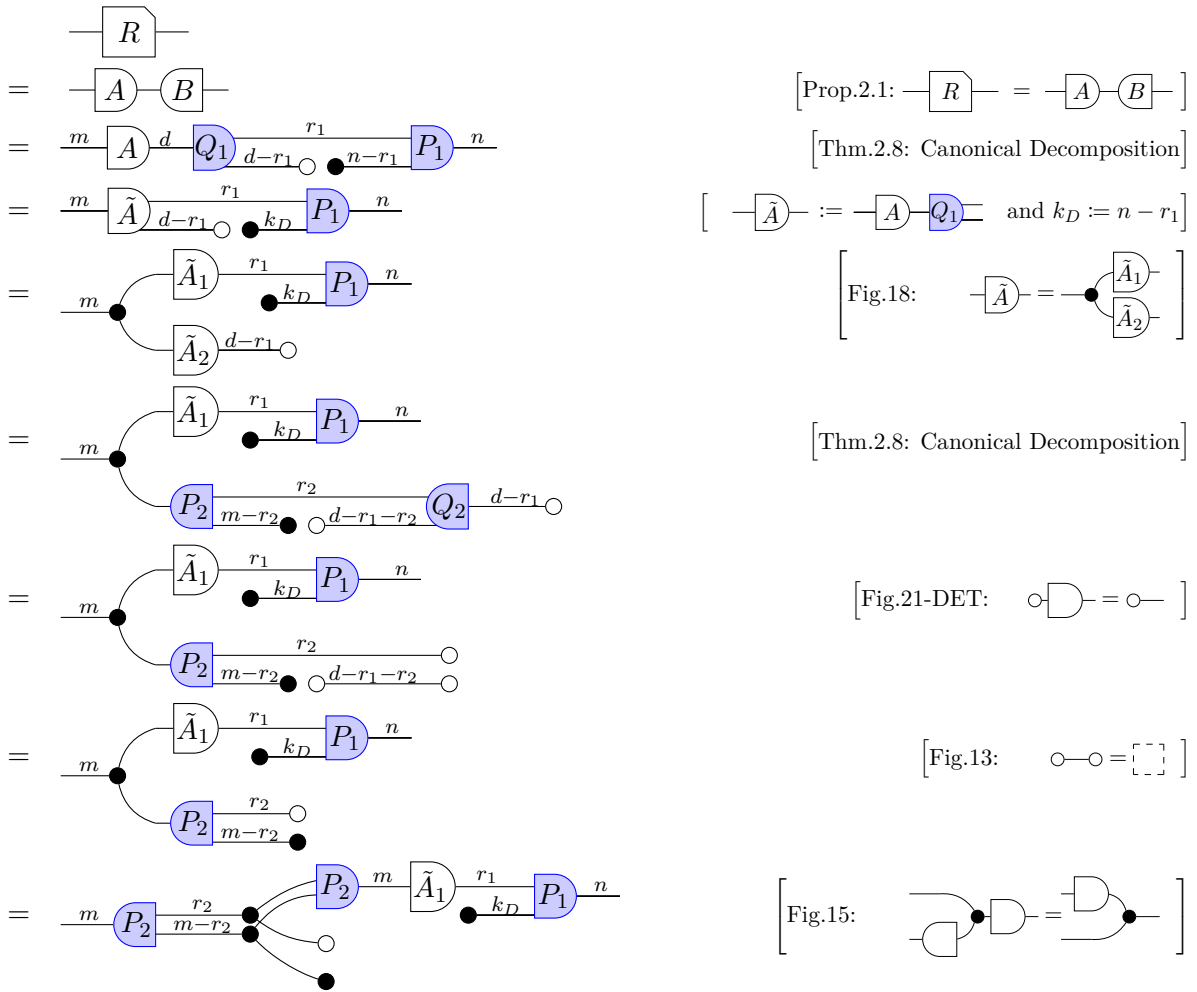
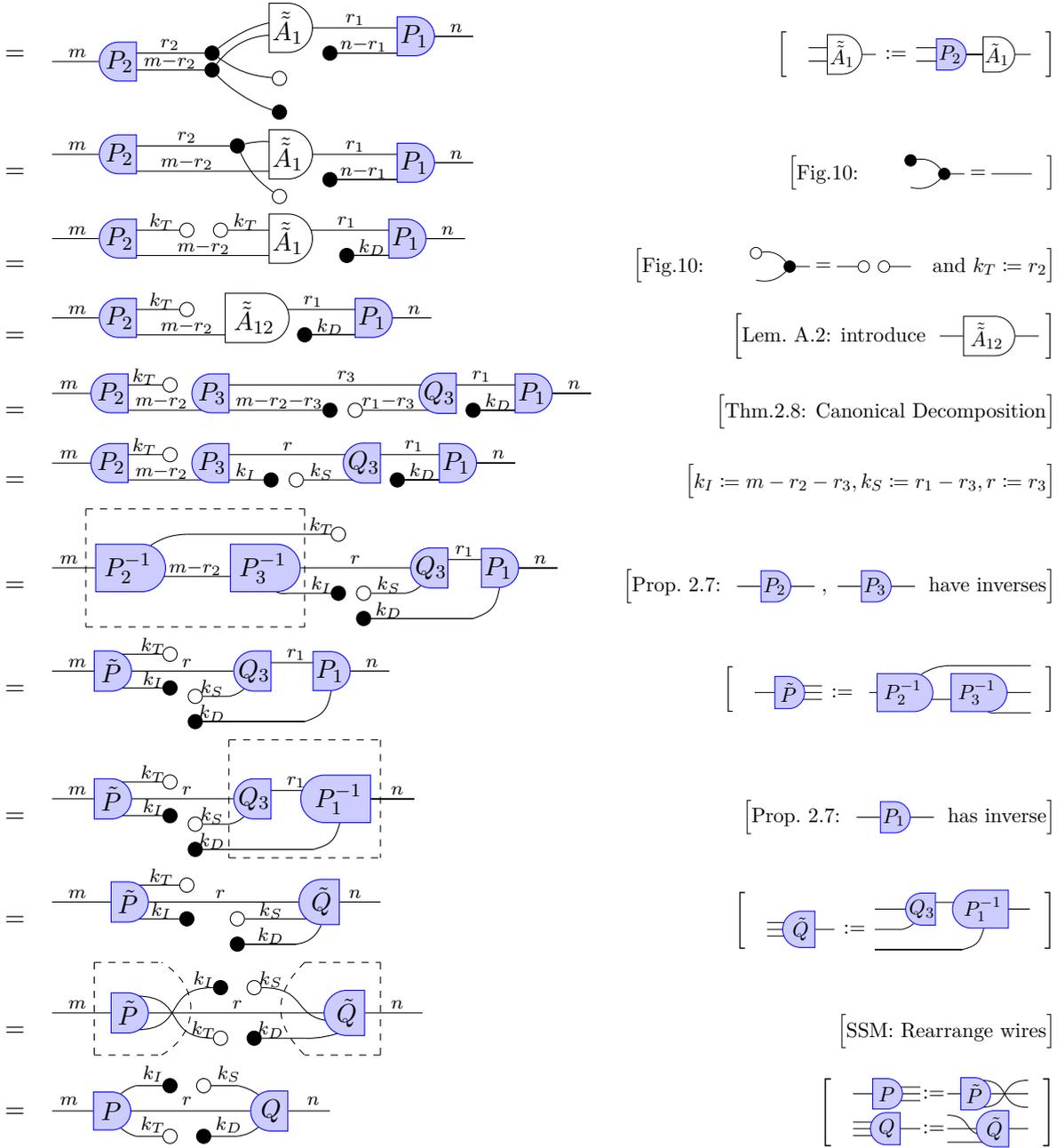
\[ \Rel{R} = \GRR.\]
\end{theorem}

\begin{proof}
\begin{hcalculation}[=]{\Rel{R}}
  \hstep{\CoSpan{A}{B}}{Prop.\ref{prop:cospan}:$\Rel{R} = \CoSpan{A}{B}$}
  \hstep{\input{tikz/figures/rcd_for_cospan/step1.tikz}}{Thm.\ref{thm:gaussian_elimination}: Canonical Decomposition}
  \hstep{\input{tikz/figures/rcd_for_cospan/step2.tikz}}{\input{tikz/figures/rcd_for_cospan/defAtil.tikz} and $k_D \coloneq n - r_1$}
  \hstep{\input{tikz/figures/rcd_for_cospan/step3.tikz}}{Fig.\ref{fig:Split}: \input{tikz/figures/theorems/splitblack2.tikz}}
  \hstep{\input{tikz/figures/rcd_for_cospan/step4.tikz}}{Thm.\ref{thm:gaussian_elimination}: Canonical Decomposition}
  \hstep{\input{tikz/figures/rcd_for_cospan/step41.tikz}}{Fig.\ref{fig:typerelations}-DET: \input{tikz/figures/theorems/whiteball.tikz}}
  \hstep{\input{tikz/figures/rcd_for_cospan/step5.tikz}}{Fig.\ref{fig:Frobenius_Algebra}: \input{tikz/figures/theorems/bone-white.tikz}}
  \hstep{\input{tikz/figures/rcd_for_cospan/step6.tikz}}{Fig.\ref{fig:Flip}: \input{tikz/figures/theorems/flipblack.tikz}}
  \hstep{\input{tikz/figures/rcd_for_cospan/step7.tikz}}{\input{tikz/figures/rcd_for_cospan/defAtiltil.tikz}}
  \hstep{\input{tikz/figures/rcd_for_cospan/step71.tikz}}{Fig.\ref{fig:Commutative_Comonoid}: \input{tikz/figures/theorems/unitblack.tikz}}
  \hstep{\input{tikz/figures/rcd_for_cospan/step8.tikz}}{Fig.\ref{fig:Commutative_Comonoid}: \input{tikz/figures/theorems/whitedestroyer2.tikz}  and $k_T \coloneq r_2$}
  \hstep{\input{tikz/figures/rcd_for_cospan/step9.tikz}}{Lem.~\ref{lem:break_a}: introduce $\Map{\tilde{\tilde{A}}_{12}}$}
  \hstep{\input{tikz/figures/rcd_for_cospan/step10.tikz}}{Thm.\ref{thm:gaussian_elimination}: Canonical Decomposition}
  \hstep{\input{tikz/figures/rcd_for_cospan/step101.tikz}}{$k_I \coloneq m-r_2-r_3, k_S \coloneq r_1 - r_3, r \coloneq r_3$}
  \hstep{\input{tikz/figures/rcd_for_cospan/step11.tikz}}{Prop.~\ref{thm:prop_inverse}: $\Inv{P_2}, \Inv{P_3}$ have inverses}
    \hstep{\input{tikz/figures/rcd_for_cospan/step112.tikz}}{\input{tikz/figures/rcd_for_cospan/defPtil.tikz}}
    \hstep{\input{tikz/figures/rcd_for_cospan/step113.tikz}}{Prop.~\ref{thm:prop_inverse}: $\Inv{P_1}$ has inverse}
  \hstep{\input{tikz/figures/rcd_for_cospan/step12.tikz}}{\input{tikz/figures/rcd_for_cospan/defQtil.tikz}}
  \hstep{\input{tikz/figures/rcd_for_cospan/step13.tikz}}{\hyperref[fig:lawsSMC]{SSM}: Rearrange wires}
  \hstep{\GRR}{\input{tikz/figures/rcd_for_cospan/defP.tikz} \\ \input{tikz/figures/rcd_for_cospan/defQ.tikz}}
\end{hcalculation}
\end{proof}

Notice that whenever $\Rel{R}$ is a map,
this decomposition specializes to the Canonical Decomposition from Theorem~\ref{thm:gaussian_elimination}.
Therefore, we can view it as its relational generalization.

\subsection{Characterization of Fundamental Properties in Cospan Form}
\label{sec:applications_characterization}

With a relational decomposition, we can complete the characterization of fundamental properties mentioned in Section~\ref{sec:fundamental_properties}.
More concretely, the decomposition unveils the equivalence between Propositions~\ref{prop:r_block} and~\ref{prop:s_block}.
We then build a graph in Section~\ref{sec:aplications_poset_of_fundamental_properties} that relates the fundamental properties to each other,
followed by Section~\ref{sec:applications_characterization} on applications, i.e., the Rank Lemma and the connection between the pigeonhole principle and the Exchange Lemma.
This completes the generalization of the Invertible Matrix Theorem.
Section~\ref{sec:applications_invertible_matrix_theorem} shows how the usual IMT is obtained as a particular case.

The next step in the characterization is connecting the abstract notions of injectivity, surjectivity, determinism, and totality to a structural property: the vanishing of wires after decomposing the relation.

\begin{lemma}
\label{lem:no_lolipop}
Let \Rel{R} be decomposed as in Theorem~\ref{thm:rcd_for_cospan}, i.e.,
\[ \Rel{R} = \GRR.\]
The vanishing of central wires determines its fundamental properties.
\begin{center}
\renewcommand{\arraystretch}{2.0}
\begin{tabular}{lcr|lcr}
        \Rel{R} is INJ   & $\iff$ & $k_I = 0$
        & \Rel{R} is SUR & $\iff$ & $k_S = 0$  \\
        \Rel{R} is TOT   & $\iff$ & $k_T = 0$
        & \Rel{R} is DET & $\iff$ & $k_D = 0$
\end{tabular}
\end{center}
\end{lemma}

\begin{proof} We prove the TOT case. The other three cases are proven similarly.
\begin{hcalculation}[\iff]{\Discard \subseteq \Diagram{Rel}{R}{.}{Black}}
\hstep{\input{tikz/figures/no_lolipop/step2.tikz}}{Thm.~\ref{thm:rcd_for_cospan}: Cospan Decomposition}
\hstep{\input{tikz/figures/no_lolipop/step3.tikz}}{Fig.\ref{fig:typerelations}-SUR: \input{tikz/figures/theorems/blackballinv.tikz}}
\hstep{\input{tikz/figures/no_lolipop/step4.tikz}}{Fig.\ref{fig:Frobenius_Algebra}: \input{tikz/figures/theorems/bone.tikz} $+$ \ref{fig:Bialgebra}: \input{tikz/figures/theorems/bonecolor.tikz}}
\hstep{\input{tikz/figures/no_lolipop/step5.tikz}}{Prop.~\ref{thm:prop_inverse}: $\Inv{P}$ has inverse}
\hstep{\input{tikz/figures/no_lolipop/step6.tikz}}{Fig.\ref{fig:typerelations}-SUR: \input{tikz/figures/theorems/blackballinv.tikz}}
\hstep{\TypedDiscard{k_3} \subseteq \TypedCoZero{k_3}}{Remove \Discard from both sides}
\hstep{k_T = 0}{Fig.\ref{fig:Inequality}: $\TypedDiscard{n} \subseteq\TypedCoZero{n} \iff n=0$}
\end{hcalculation}
\end{proof}

Now we are ready to prove that all statements in Lemma~\ref{lem:no_lolipop}
and Propositions~\ref{prop:r_block} and~\ref{prop:s_block} are equivalent.
This is a generalization of the Invertible Matrix Theorem
giving a series of alternative characterization for each fundamental property.

\begin{theorem}[Characterization of fundamental properties]\label{thm:cospan_dict}
  Let \CoSpan{A}{B} be decomposed as in Theorem~\ref{thm:rcd_for_cospan}, i.e.,
  \[ \CoSpan{A}{B} = \GRR.\]
  Each column of the following table is a series of equivalent statements.
  \begin{center}
    \resizebox{\textwidth}{!}{\input{tikz/figures/thm_cospan_dict/stat_1.tikz}}
  \end{center}
\end{theorem}
\begin{proof}
  We only prove the TOT case (second column), the other can be proven similarly.
  The equivalence between rows~1,~5, and~6
  is the subject of Proposition~\ref{prop:r_block},
  whereas rows~3 and~4 are equivalent by Proposition~\ref{prop:s_block}.
  Therefore, it only remains to show that rows~1,~2 and~3 are equivalent.
  The strategy is Row 1 $\Rightarrow$ Row 2 $\Rightarrow$ Row 3 $\Rightarrow$ Row 1.

  (Row 1 $\implies$ Row 2) Lemma~\ref{lem:no_lolipop}.

  (Row 2 $\implies$ Row 3)
    \begin{hcalculation}[=]{\CoSpan{A}{B}}
      \hstep{\GRR}{Thm.~\ref{thm:rcd_for_cospan}: Cospan Decomposition}
      \hstep{\input{tikz/figures/thm_cospan_dict/2_implies_3_step_3.tikz}}{Hyp: $k_T = 0$}
      \hstep[\supseteq]{\input{tikz/figures/thm_cospan_dict/2_implies_3_step_4.tikz}}{Fig.\ref{fig:Inequality}: $\Zero \subseteq \CoDiscard$}
      \hstep[\supseteq]{\input{tikz/figures/thm_cospan_dict/2_implies_3_step_5.tikz}}{Prop.~\ref{thm:prop_inverse}: $\Inv{Q}$ has inverse}
      \hstep{\Map{S}}{$\Map{S} \coloneq
      \input{tikz/figures/thm_cospan_dict/2_implies_3_step_4.tikz}$}
    \end{hcalculation}

  (Row 3 $\implies$ Row 1)
    \begin{hcalculation}[\supseteq]{\CompTwo{., Map, A, CoMap, B, Black}}
      \hstep{\Diagram{Map}{S}{.}{Black}}{Hyp: $\CoSpan{A}{B} \supseteq \Map{S}$}
      \hstep[=]{\Discard}{Fig.\ref{fig:typerelations}-TOT: $\Diagram{Map}{}{.}{Black} = \Discard$}
    \end{hcalculation}
\end{proof}

\begin{remark}
  Theorem~\ref{prop:s_block} only proved existence results,
  without a procedure to construct the maps $\Map{S}$ relating the cospan terms.
  In contrast to that, the previous theorem uses the relational decomposition
  to explicitly construct the maps.
\end{remark}

\subsection{Graph of Fundamental Properties}
\label{sec:aplications_poset_of_fundamental_properties}

We now discuss how the fundamental properties relate to each other.
This result can also be visualized as the graph in Figure~\ref{fig:poset_of_fundamental_properties}.

\begin{figure}[thb]
  \centering
  \input{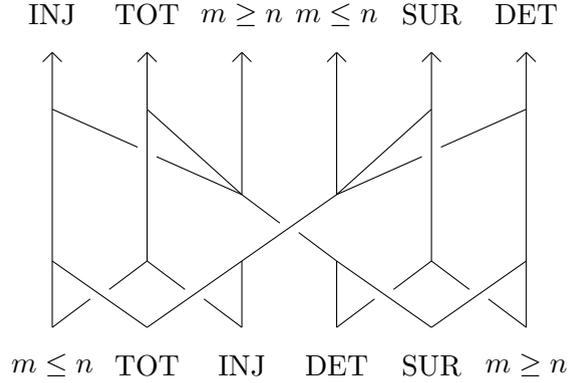}
  \caption{Graph of fundamental properties.
   The way to read the graph is as follows: given a node $Y$ at the top, let $X_1, X_2, ..., X_k$ be all the nodes at the bottom you can reach by following the paths starting at $Y$. Then, it holds that $X_1, ..., X_k \implies Y$. For example, if we start at the node $(m \leq n)$ at the top and follow the edges downwards, we reach the nodes TOT and INJ at the bottom, so it holds that TOT + INJ $\implies m \leq n$, which is statement (v) of Theorem~\ref{thm:poset_of_fundamental_properties}.}
  \label{fig:poset_of_fundamental_properties}
\end{figure}

\begin{theorem}
  \label{thm:poset_of_fundamental_properties}
  For a relation $\Rel{R}$, the following hold\newline
\begin{tabular}{rll}
(i)   & $R$ is TOT, INJ, SUR and $m \geq n$ & $\implies R$ is DET; \\
(ii)  & $R$ is DET, INJ, SUR and $m \leq n$ & $\implies R$ is TOT; \\
(iii) & $R$ is DET, TOT, INJ and $m \geq n$ & $\implies R$ is SUR; \\
(iv)  & $R$ is DET, TOT, SUR and $m \leq n$ & $\implies R$ is INJ; \\
(v)   & $R$ is TOT and INJ                  & $\implies m \leq n$; \\
(vi)  & $R$ is DET and SUR                  & $\implies m \geq n$.
\end{tabular}
\end{theorem}
\begin{proof} We prove (i) and (iii). The other arguments are analogous.

(i)
\begin{hcalculation}[=]{\TypedRel{R}{m}{n}}
\hstep{\GRR}{Thm.~\ref{thm:rcd_for_cospan}: Cospan Decomposition}
\hstep{\input{tikz/figures/poset_of_fundamental_properties/step3.tikz}}{Lem.~\ref{lem:no_lolipop}: $k_S = k_T = k_D = 0$}
\hstep[\implies]{n + k_I = m}{$\Inv{P}$ and $\Inv{Q}$ are square}
\hstep[\implies]{k_I = 0}{$m \leq n$ and $k_I \geq 0$}
\hstep[\implies]{\Rel{R} \text{ is INJ}}{Lem.~\ref{lem:no_lolipop}}
\end{hcalculation}

(iii)
\begin{hcalculation}[=]{\TypedRel{R}{m}{n}}
\hstep{\GRR}{Thm.~\ref{thm:rcd_for_cospan}: Cospan Decomposition}
\hstep{\input{tikz/figures/poset_of_fundamental_properties/step8.tikz}}{Lem.~\ref{lem:no_lolipop}: $k_I = k_T = 0$}
\hstep[\implies]{m = r \text{ and } r + k_T + k_D = n}{$\Inv{P}$ and $\Inv{Q}$ are square}
\hstep[\implies]{m + k_T + k_D = n}{Substitute $m = r$}
\hstep[\implies]{m \leq n}{$k_T, k_D \geq 0$}
\end{hcalculation}
\end{proof}

We are ready to state the equivalence between the Pigeonhole Principle and the Exchange Lemma discussed in Section~\ref{sec:pigeonhole}. As it is usually formulated, the Exchange Lemma states that if $A \subseteq \K^k$ is a set of $m$ linearly independent vectors and $B \subseteq \K^k$ is a set of $n$ spanning vectors, then $m \leq n$. As a slightly cleaner way of saying this, we may arrange the vectors in $A$ (resp. $B$) as columns of a matrix. We arrive at the following.
\begin{center}
\renewcommand{\arraystretch}{2.0}
\begin{tabular}{r l}
    (Exchange Lemma) &$\TypedMap{A}{m}{k}$ is INJ and $\;\TypedImage{B}{n}{k}\!\!\! \supseteq \Image{A} \implies m \leq n$. \\
    (Pigeonhole Principle) &$\TypedCoSpan{A}{B}{m}{k}{n}$ is TOT and INJ $ \implies m \leq n$.
\end{tabular}
\end{center}
We proceed to prove that their hypotheses are equivalent.
\begin{proposition}
Let $\TypedMap{A}{m}{k}, \TypedMap{B}{n}{k}$ be maps, then, the following are equivalent.
\begin{enumerate}
    \item[(i)] $\Map{A}$ is INJ and $\;\Image{B} \!\!\!\supseteq \Image{A}$,
    \item[(ii)] $\CoSpan{A}{B}$ is TOT and INJ.
\end{enumerate}
\end{proposition}
\begin{proof} This is a simple application of the characterization in Theorem~\ref{thm:cospan_dict}.
\end{proof}

\subsection{Invertible Matrix Theorem}
\label{sec:applications_invertible_matrix_theorem}

The \emph{Invertible Matrix Theorem} (IMT) is a central result in linear algebra
enumerating equivalent characterizations for a square matrix to have an inverse.
Theorem~\ref{thm:cospan_dict} characterizes a relation according to its fundamental properties,
and functions as a relational generalization of the IMT.
This section recovers a characterization of matrices with inverse
from Theorem~\ref{thm:cospan_dict} by applying it to the map $\Map{A}$
viewed as a cospan, i.e., with $\CoMap{B} = \Id$.
This way, the theorem specializes to a characterization of left and right-invertibility.
Following Section~\ref{sec:introduction},
we first characterize injective and surjective maps separately
and then link both results for square maps.

\begin{theorem}[Surjective/Injective Matrix Theorem]
\label{thm:ivm_sur_inj}
For a linear map \Map{A},
each column of the following table is a series of equivalent statements.
\begin{center}
\input{tikz/figures/thm_surjective_injective_matrix/stat_1.tikz}
\end{center}
\end{theorem}
\begin{proof}
  Apply Theorem~\ref{thm:cospan_dict} to the cospan \Map{A}.
\end{proof}

For square matrices, all these statements are, in fact, equivalent.
This next result is the Invertible Matrix Theorem,
and comes as a consequence of the graph of fundamental properties.

\begin{theorem}[Invertible Matrix Theorem]
\label{thm:ivm_inv}
For a square map $\TypedMap{A}{n}{n}$,
all statements in Theorem~\ref{thm:ivm_sur_inj} are equivalent.
\end{theorem}
\begin{proof}
  From Theorem~\ref{thm:poset_of_fundamental_properties}, $\Map{A}$ is INJ if and only if it is SUR.
\end{proof}

Finally, it is worth noticing that, as corollary,
the IMT provides converse to Propositions~\ref{thm:prop_linverse},~\ref{thm:prop_rinverse} and~\ref{thm:prop_inverse}.
While those propositions conclude a fundamental property from the maps' invertibility,
we can now assure the existence of an inverse from these properties.

\begin{corollary}[Existence of Inverse]
  Consider a map $\Map{A}$. From the IMT, it is\newline
    \begin{tabular}{rrcl}
      (i)   & injective  &$\iff$ &it has a left-inverse; \\
      (ii)  & surjective &$\iff$ &it has a right-inverse; \\
      (iii) & bijective  &$\iff$ &it has an inverse.
    \end{tabular}
\end{corollary}

\section{Pair of Matrices and their Relevant Subspaces}

The Generalized Singular Value Decomposition (GSVD)~\cite{gsvd1976,gsvd1981}
is a powerful tool for simultaneously decomposing a pair of matrices.
Unfortunately, similarly to the usual SVD, it only works for the fields of real and complex numbers.
In this section, we show that the relational decomposition from Theorem~\ref{thm:rcd_for_cospan}
induces a similar simultaneous matrix decomposition that works without assumptions on the underlying scalar field.
We finish this section by showing how this decomposition can compute the main subspaces involving any pair of linear maps.

\subsection{Decomposing Pairs of Matrices}
\label{sec:decomposition_matrices}

Theorem~\ref{thm:rcd_for_cospan} provides a decomposition for any cospan $\CoSpan{A}{B}$.
Since these are formed by two linear maps, a question arises: can we turn it into mutual---but separate---decompositions for its constituent factors $\Map{A}$ and $\Map{B}$?
The answer is positive, requiring the introduction of another linear map $\Map{H}$ connecting both factorizations. We construct the needed map in Theorem~\ref{thm:abstract} and associate it with the linear decomposition in Theorem~\ref{thm:rcd_for_pairs_of_matrices}.

The next theorem is this section's workhorse
in that, it constructs a linking relation with all the properties needed for the main result. It presents well-known hypothesis when dealing with spans and cospans~\cite{freyd1990categories}.

\begin{theorem}\label{thm:abstract}
For all $\CoSpan{A}{B}=\CoSpan{D_1}{D_2}$
where \input{tikz/figures/rcd_for_pairs_of_matrices/TheoremAbstract/D1D2.tikz} is surjective,
let
\begin{equation}
\Rel{H}=\input{tikz/figures/rcd_for_pairs_of_matrices/TheoremAbstract/defH.tikz}
\end{equation}
The relation just defined satisfies
\begin{enumerate}
    \item[(i)] \Rel{H} is deterministic;
    \item[(ii)] \Rel{H} is total;
    \item[(iii)] \Rel{H} is injective;
    \item[(iv)] \Rel{H} is surjective if \input{tikz/figures/rcd_for_pairs_of_matrices/TheoremAbstract/AB.tikz} is surjective;
    \item[(v)] \input{tikz/figures/rcd_for_pairs_of_matrices/TheoremAbstract/defAeB_Rel.tikz};
    \item[(vi)] \Rel{H} is the unique linear relation satisfying (i), (ii) and (v).
\end{enumerate}
\end{theorem}

\begin{proof}
(i) \Rel{H} is deterministic.
\begin{hcalculation}[=]{\input{tikz/figures/rcd_for_pairs_of_matrices/TheoremAbstract/item1/hSV_step1.tikz}}
\hstep{\input{tikz/figures/rcd_for_pairs_of_matrices/TheoremAbstract/item1/hSV_step2.tikz}}{Def. $\Rel{H}$}
\hstep{\input{tikz/figures/rcd_for_pairs_of_matrices/TheoremAbstract/item1/hSV_step3.tikz}}{Lemma~\ref{lemma:abstract}}
\hstep[\subseteq]{\input{tikz/figures/rcd_for_pairs_of_matrices/TheoremAbstract/item1/hSV_step4.tikz}}{Fig.\ref{fig:typerelations}-DET: \input{tikz/figures/theorems/detid.tikz}}
\hstep{\Zero}{Fig.\ref{fig:Frobenius_Algebra}: \input{tikz/figures/theorems/frobenius.tikz}}
\end{hcalculation}

(ii) \Rel{H} is total.
\begin{hcalculation}[=]{\input{tikz/figures/rcd_for_pairs_of_matrices/TheoremAbstract/item1/hTOT_step1.tikz}}
\hstep{\input{tikz/figures/rcd_for_pairs_of_matrices/TheoremAbstract/item1/hTOT_step2.tikz}}{Def. $\Rel{H}$}
\hstep{\input{tikz/figures/rcd_for_pairs_of_matrices/TheoremAbstract/item1/hTOT_step3.tikz}}{Fig.\ref{fig:Bialgebra}: \input{tikz/figures/theorems/blackdestroyer.tikz}}
\hstep{\input{tikz/figures/rcd_for_pairs_of_matrices/TheoremAbstract/item1/hTOT_step4.tikz}}{Fig.\ref{fig:typerelations}-TOT: \input{tikz/figures/theorems/blackball.tikz}}
\hstep[\supseteq ]{\input{tikz/figures/rcd_for_pairs_of_matrices/TheoremAbstract/item1/hTOT_step5.tikz}}{Fig.\ref{fig:typerelations}-SUR: \input{tikz/figures/theorems/sur.tikz}}
\end{hcalculation}

(iii) \Rel{H} is injective.
\begin{hcalculation}[=]{\input{tikz/figures/rcd_for_pairs_of_matrices/TheoremAbstract/item4/step1.tikz}}
\hstep{\input{tikz/figures/rcd_for_pairs_of_matrices/TheoremAbstract/item4/step2.tikz}}{Def. $\Rel{H}$}
\hstep{\input{tikz/figures/rcd_for_pairs_of_matrices/TheoremAbstract/item4/step3.tikz}}{Lemma~\ref{lemma:abstract}}
\hstep{\input{tikz/figures/rcd_for_pairs_of_matrices/TheoremAbstract/item4/step4.tikz}}{Fig.\ref{fig:typerelations}-SUR: \input{tikz/figures/theorems/surideq.tikz}}
\hstep{\CoZero}{Fig.\ref{fig:Frobenius_Algebra}: \input{tikz/figures/theorems/frobenius.tikz}}
\end{hcalculation}

(iv) If \input{tikz/figures/rcd_for_pairs_of_matrices/TheoremAbstract/AB.tikz} is surjective, then so is \Rel{H}.
\begin{hcalculation}[=]{\input{tikz/figures/rcd_for_pairs_of_matrices/TheoremAbstract/item5/step1.tikz}}
\hstep{\input{tikz/figures/rcd_for_pairs_of_matrices/TheoremAbstract/item5/step2.tikz}}{Def. $\Rel{H}$}
\hstep{\input{tikz/figures/rcd_for_pairs_of_matrices/TheoremAbstract/item5/step3.tikz}}{Fig.\ref{fig:Bialgebra}: \input{tikz/figures/theorems/blackdestroyer.tikz}}
\hstep{\input{tikz/figures/rcd_for_pairs_of_matrices/TheoremAbstract/item5/step4.tikz}}{Fig.\ref{fig:typerelations}-TOT: \input{tikz/figures/theorems/blackball.tikz}}
\hstep{\CoDiscard}{Hyp: \ref{fig:typerelations}-SUR: \input{tikz/figures/theorems/sureq.tikz}}
\end{hcalculation}

(v) We show $\Map{A} \subseteq \Comp{D_1/Map, H/Rel}$.
The proof for $\Map{B}$ is similar.
The equalities follow from Lemma~\ref{lem:pair_of_matrices_1} and items~(i) and~(ii) of this theorem.
\begin{hcalculation}[=]{\input{tikz/figures/rcd_for_pairs_of_matrices/TheoremAbstract/item2/step1.tikz}}
\hstep{\input{tikz/figures/rcd_for_pairs_of_matrices/TheoremAbstract/item2/step2.tikz}}{Def. $\Rel{H}$}
\hstep[\supseteq]{\input{tikz/figures/rcd_for_pairs_of_matrices/TheoremAbstract/item2/step3.tikz}}{Fig.\ref{fig:Minimum_and_maximum}: \input{tikz/figures/theorems/zero.tikz}}
\hstep{\input{tikz/figures/rcd_for_pairs_of_matrices/TheoremAbstract/item2/step4.tikz}}{Fig.\ref{fig:Commutative_Comonoid}: \input{tikz/figures/theorems/unit.tikz}}
\hstep[\supseteq]{\Map{A}}{Fig.\ref{fig:typerelations}-TOT: \input{tikz/figures/theorems/totalid1.tikz}}
\end{hcalculation}

(vi) Lastly, we show $\Rel{H}$ is the unique map satisfying (v).
Suppose we have two maps $\Map{H_1}, \Map{H_2}$ satisfying (v). Then,
\begin{hcalculation}[\iff]{\input{tikz/figures/rcd_for_pairs_of_matrices/TheoremAbstract/item3/step1.tikz}}
\hstep{\input{tikz/figures/rcd_for_pairs_of_matrices/TheoremAbstract/item3/step12.tikz}}{Fig.\ref{fig:Connect_sum}: \input{tikz/figures/theorems/connect.tikz}}
\hstep{\input{tikz/figures/rcd_for_pairs_of_matrices/TheoremAbstract/item3/step2.tikz}}{Fig.\ref{fig:Flip}: \input{tikz/figures/theorems/flip2.tikz}}
\hstep[\implies]{\input{tikz/figures/rcd_for_pairs_of_matrices/TheoremAbstract/item3/step3.tikz}}{Multiply \input{tikz/figures/rcd_for_pairs_of_matrices/TheoremAbstract/D1D21.tikz}}
\hstep{\input{tikz/figures/rcd_for_pairs_of_matrices/TheoremAbstract/item3/step5.tikz}}{Fig.\ref{fig:typerelations}-SUR: \input{tikz/figures/theorems/surideq.tikz}}
\hstep{\input{tikz/figures/rcd_for_pairs_of_matrices/TheoremAbstract/item3/step4.tikz}}{Fig.\ref{fig:Frobenius_Algebra}: \input{tikz/figures/theorems/frobenius.tikz}}
\end{hcalculation}
This concludes the proof.
\end{proof}

\begin{remark}
  Theorem~\ref{thm:abstract} implies that the relation $\Rel{H}$ is a linear map.
  Therefore, we from now on denote it as $\Map{H}$.
\end{remark}

We proceed to show how to attain a decomposition for a pair $\Map{A}$ and $\Map{B}$ given the relational decomposition of the cospan form $\CoSpan{A}{B}$. For this, we need Definition~\ref{def:D1_e_D2} and Lemma~\ref{lem:D1_D2sur} below.

\begin{definition}
  \label{def:D1_e_D2}
  For maps $\Map{A}$, $\Map{B}$ decomposed as in Theorem~\ref{thm:rcd_for_cospan},
  \[\TypedRel{R}{m}{n} = \GRR,\]
  define the auxiliary maps $\Map{D_1}$ and $\Map{D_2}$ as
  \[
    \begin{aligned}
      \input{tikz/figures/rcd_for_pairs_of_matrices/defD1.tikz} & &
      \input{tikz/figures/rcd_for_pairs_of_matrices/defD2.tikz}
    \end{aligned}
  \]
\end{definition}

\begin{lemma}\label{lem:D1_D2sur} The maps \Map{D_1} and \Map{D_2} of Definition~\ref{def:D1_e_D2} are such that $\input{tikz/figures/rcd_for_pairs_of_matrices/step22.tikz}$ is surjective.
\end{lemma}
\begin{proof}
\begin{hcalculation}[=]{\input{tikz/figures/rcd_for_pairs_of_matrices/step2.tikz}}
\hstep{\input{tikz/figures/rcd_for_pairs_of_matrices/step3.tikz}}{Def.\ref{def:D1_e_D2}: $\Map{D_1}$ and $\Map{D_2}$}
\hstep{\CoDiscard}{Fig.\ref{fig:Frobenius_Algebra}: \input{tikz/figures/theorems/bone.tikz} $+$ \ref{fig:Commutative_Comonoid}: \input{tikz/figures/theorems/unit.tikz}}
\end{hcalculation}
\end{proof}

We have all the necessary tools for simultaneously decomposing a matrix pair.
The next results puts together Theorems~\ref{thm:rcd_for_cospan} and~\ref{thm:abstract}
to construct the desired matrix decompositions.

\begin{theorem}[Pair Decomposition]
  \label{thm:rcd_for_pairs_of_matrices}
Let $\Rel{R}=\CoSpan{A}{B}$ be decomposed as in Theorem~\ref{thm:rcd_for_cospan}, i.e.,
\[\TypedRel{R}{m}{n} = \GRR.\]
Then, there exists a unique injective linear map $\Map{H}$ such that
\[\input{tikz/figures/rcd_for_pairs_of_matrices/stat.tikz}\]
\end{theorem}

\begin{proof} Note that
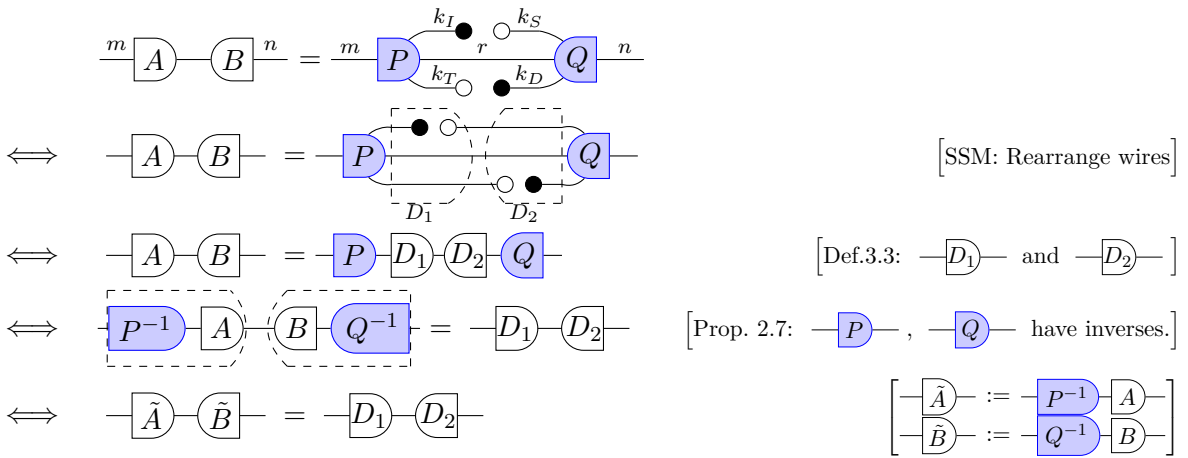
\begin{hcalculation}[\iff]{\TypedCoSpan{A}{B}{m}{}{n} = \GRR}
\hstep{\CoSpan{A}{B} = \input{tikz/figures/rcd_for_pairs_of_matrices/step1.tikz}}{\hyperref[fig:lawsSMC]{SSM}: Rearrange wires}
\hstep{\CoSpan{A}{B}=\input{tikz/figures/rcd_for_pairs_of_matrices/step12.tikz}}{Def.\ref{def:D1_e_D2}: $\Map{D_1}$ and $\Map{D_2}$}
\hstep{\input{tikz/figures/rcd_for_pairs_of_matrices/step11.tikz} = \CoSpan{D_1}{D_2}}{Prop.~\ref{thm:prop_inverse}: $\Inv{P}$, $\Inv{Q}$ have inverses.}
\hstep{\CoSpan{\tilde{A}}{\tilde{B}} = \CoSpan{D_1}{D_2}}{$\input{tikz/figures/rcd_for_pairs_of_matrices/defAtil.tikz}$ \\ $\input{tikz/figures/rcd_for_pairs_of_matrices/defBtil.tikz}$}
\end{hcalculation}

By Lemma~\ref{lem:D1_D2sur} and Theorem~\ref{thm:abstract},
there exists a unique injective linear map
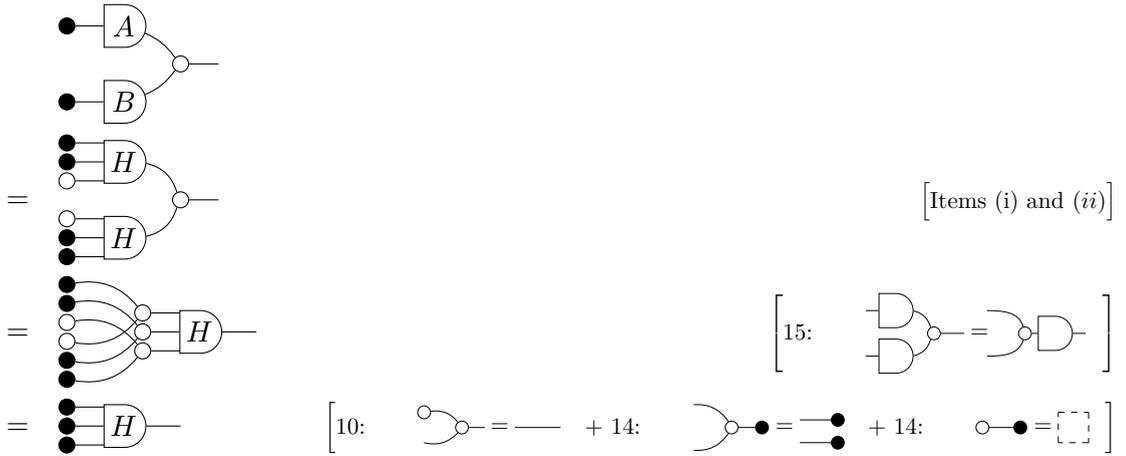
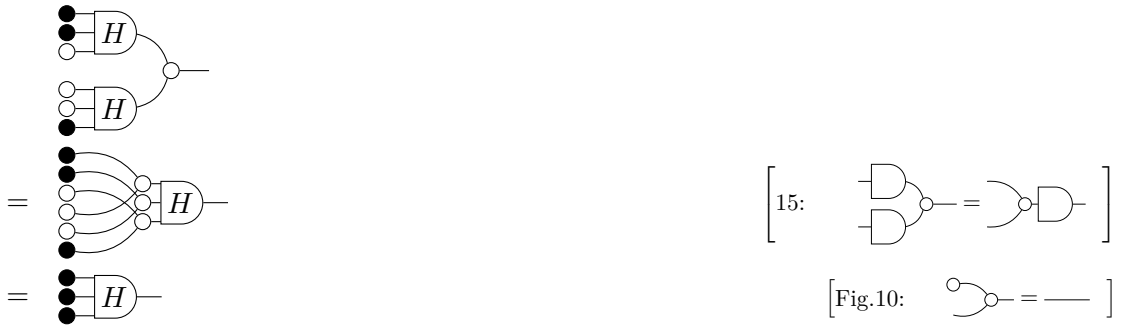
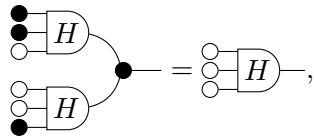
\[\Map{H}=\input{tikz/figures/rcd_for_pairs_of_matrices/defH.tikz},\]
and, from Theorem~\ref{thm:abstract} item (v),
\begin{hcalculation}[\iff]{\input{tikz/figures/rcd_for_pairs_of_matrices/step5.tikz}}
\hstep{\input{tikz/figures/rcd_for_pairs_of_matrices/step6.tikz}}{Def. $\Map{\tilde{A}}$ and $\Map{\tilde{B}}$}
\hstep{\input{tikz/figures/rcd_for_pairs_of_matrices/step7.tikz}}{Multiply $\input{tikz/figures/rcd_for_pairs_of_matrices/step9.tikz}$}
\hstep{\input{tikz/figures/rcd_for_pairs_of_matrices/step8.tikz}}{Prop.~\ref{thm:prop_inverse}: $\Inv{P}$ and $\Inv{Q}$ have inverses}
\hstep{\input{tikz/figures/rcd_for_pairs_of_matrices/stat.tikz}}{Def. $\Map{D_1}$ and $\Map{D_2}$}
\end{hcalculation}
\end{proof}

\subsection{Calculating Subspaces}
\label{sec:applications_subspaces}
For a pair of matrices $A$ and $B$ with the same amount of rows,
a common problem in linear algebra is to find bases for the sum $\im(A) + \im(B)$ and the intersection $\im(A) \cap \im(B)$.
One method to calculate them is the Zassenhaus Algorithm~\cite{Fischer2012}. It consists of putting the matrix $\begin{bsmallmatrix} A & A \\ B & 0 \end{bsmallmatrix}$ in row-echelon form. The bases are then constructed from the coefficients of the resulting matrix. In this section, we show that the matrix $H$ obtained in Theorem~\ref{thm:abstract} is an alternative method for calculating such bases. Its columns contain not only bases for the sum and intersection but also bases for other relevant subspaces obtained from $\im(A)$ and $\im(B)$ using the standard operations on vector spaces---sum, intersection, and complement, as defined below.

\begin{definition}
    Given subspaces $\Subspace{A}$, $\Subspace{B}$,
    define their
    \begin{enumerate}
      \item[(i)] \textbf{Sum:}
          $\input{tikz/figures/subspaces/AB_sum.tikz}      = \left\{a+b \mid a \in A \;\text{and}\; b \in B \right\}$,
        \item[(ii)] \textbf{Intersection:}
          $\input{tikz/figures/subspaces/AB_intersec.tikz} = \left\{x \mid x \in A \;\text{and}\; x \in B \right\}$,
        \item[(iii)] \textbf{Complement:}
          \emph{$\Subspace{A}$ complements $\Subspace{B}$ wrt a subspace $\Subspace{C}$} whenever
        their sum is $\Subspace{C}$ and their intersection is zero, i.e.,
        \[
          \Subspace{A + B} = \Subspace{C} \;\text{ and }\; \Subspace{A \cap B} = \Zero.
        \]
    \end{enumerate}
\end{definition}

\begin{theorem}\label{thm:subspaces}
Let $A, B$ be matrices and $\mathcal{A} \coloneq \im(A)$,
$\mathcal{B} \coloneq \im(B)$.
Then, we have the following equalities,
where the complements are with respect to $\mathcal{A} + \mathcal{B}$.
\begin{center}
\bgroup
\renewcommand*{\arraystretch}{2.0}
\begin{tabular}{crcl crcl}
$(i)$   & $\mathcal{A}$                  & $=$ & \input{tikz/figures/subspaces/imAH.tikz}        & $(v)$         & $\mathcal{A}$ is complemented by               & \; & \input{tikz/figures/subspaces/b-a.tikz} \\
$(ii)$  & $\mathcal{B}$                  & $=$ & \input{tikz/figures/subspaces/imBH.tikz}        & $(vi)$        & $\mathcal{B}$ is complemented by  & \;           & \input{tikz/figures/subspaces/a-b.tikz} \\
$(iii)$ & $\mathcal{A} \cap \mathcal{B}$ & $=$ & \input{tikz/figures/subspaces/intersecH.tikz}   & $(vii)$       & $\mathcal{A} \cap \mathcal{B}$ is complemented by   & \; & \input{tikz/figures/subspaces/ab-a-b.tikz} \\
$(iv)$  & $\mathcal{A}+\mathcal{B}$      & $=$ & \input{tikz/figures/subspaces/sumH.tikz} & $(viii)$ & $\left \{ 0 \right \}$    & $=$ & \input{tikz/figures/subspaces/empty.tikz}
\end{tabular}
\egroup
\end{center}
\end{theorem}

\begin{proof}
\begin{itemize}
\item[(i)]\textbf{[$\mathcal{A}$] }
\begin{hcalculation}[=]{\Image{A}}
\hstep{\input{tikz/figures/subspaces/imAstep1.tikz}}{Thm.~\ref{thm:rcd_for_pairs_of_matrices}}
\hstep{\input{tikz/figures/subspaces/imAH.tikz}}{Fig.\ref{fig:typerelations}-SUR: \input{tikz/figures/theorems/blackballinv.tikz}}\end{hcalculation}

\item[(ii)] [\textbf{$\mathcal{B}$}] Analogous to item~(i).

\item[(iii)]\textbf{[$\mathcal{A} \cap \mathcal{B}$]}
\begin{hcalculation}[=]{\input{tikz/figures/subspaces/intersec.tikz}}
\hstep{\input{tikz/figures/subspaces/intersecstep2.tikz}}{Items (i) and $(ii)$}
\hstep{\input{tikz/figures/subspaces/intersecstep4.tikz}}{Fig.\ref{fig:typerelations}-INJ: \input{tikz/figures/theorems/flip3.tikz}}
\hstep{\input{tikz/figures/subspaces/intersecstep5.tikz}}{Fig.\ref{fig:Commutative_Comonoid}: \input{tikz/figures/theorems/unitblack.tikz} $+$ \ref{fig:Bialgebra}: \input{tikz/figures/theorems/whitedestroyer.tikz} $+$ \ref{fig:Bialgebra}: \input{tikz/figures/theorems/bonecolor.tikz}}
\end{hcalculation}

\item[(iv)]\textbf{[$\mathcal{A} + \mathcal{B}$]}
\begin{hcalculation}[=]{\input{tikz/figures/subspaces/add_1.tikz}}
\hstep{\input{tikz/figures/subspaces/add_2.tikz}}{Items (i) and $(ii)$}
\hstep{\input{tikz/figures/subspaces/add_3.tikz}}{\ref{fig:Flip}: \input{tikz/figures/theorems/flip2.tikz}}
\hstep{\input{tikz/figures/subspaces/add_4.tikz}}{\ref{fig:Commutative_Comonoid}: \input{tikz/figures/theorems/unit.tikz} $+$ \ref{fig:Bialgebra}: \input{tikz/figures/theorems/blackdestroyer.tikz} $+$ \ref{fig:Bialgebra}: \input{tikz/figures/theorems/bonecolor.tikz}}
\end{hcalculation}

\item[(v)]\textbf{[Complement of $\mathcal{A}$ w.r.t. $\mathcal{A} + \mathcal{B}$]}

    \begin{hcalculation}[=]{\input{tikz/figures/subspaces/a-intersecstep1.tikz}}
\hstep{\input{tikz/figures/subspaces/a-intersecstep2.tikz}}{\ref{fig:Flip}: \input{tikz/figures/theorems/whitefront.tikz}}
\hstep{\input{tikz/figures/subspaces/sumH.tikz}}{Fig.\ref{fig:Commutative_Comonoid}:\input{tikz/figures/theorems/unit.tikz}}
\end{hcalculation}

\noindent A similar calculation demonstrates \[\input{tikz/figures/subspaces/a-intersec2step1.tikz} = \input{tikz/figures/subspaces/empty.tikz},\] proving that $\input{tikz/figures/subspaces/b-a.tikz}$ is the complement of $\mathcal{A}$ with respect to $\mathcal{A} + \mathcal{B}$.

\item[(vi)]\textbf{[Complement of $\mathcal{B} $ w.r.t. $\mathcal{A} + \mathcal{B}$]}  Analogous to item~(v).

\item[(vii)]\textbf{[Complement of $\mathcal{A}\cap \mathcal{B}$ w.r.t. $\mathcal{A} + \mathcal{B}$]} Analogous to item~(v).

\item [(vii)] \textbf{[$\left \{ 0 \right \}$]}
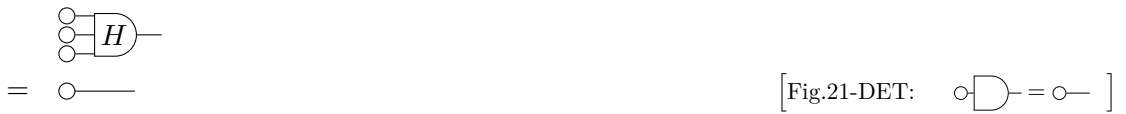
\begin{hcalculation}[=]{\input{tikz/figures/subspaces/empty.tikz}}
\hstep{\Zero}{Fig.\ref{fig:typerelations}-DET: \input{tikz/figures/theorems/whiteball.tikz}}
\end{hcalculation}
\end{itemize}
\end{proof}

From Theorem \ref{thm:subspaces}, it is possible to obtain all relevant subspaces of $\im(A)$ and $\im(B)$ through $H$.

\appendix

\section{Axioms and Theorems in Graphical Notation}
\label{sec:gla_thms}

This appendix presents the linear relations theorems in the graphical notation used throughout the paper. Only the theorems necessary for the presented proofs are recorded here. The complete presentation can be found in~\cite{PAIXAO2022}. In order to maintain the focus of this paper, most of the proofs in this section will be omitted, as they can be found in the same reference.

\paragraph*{Symmetric Strict Monoidal Categories}
\label{sec:SSM}
Raw terms are quotiented with respect to the laws of symmetric strict monoidal (SSM) categories, summarized in Fig.~\ref{fig:lawsSMC}. We omit the (well-known) details~\cite{selinger2010survey} here and mention only that this amounts to eschewing the need for ``dotted line boxes''
and ensuring that diagrams with the same topological connectivity are equated.

\begin{figure}[h]
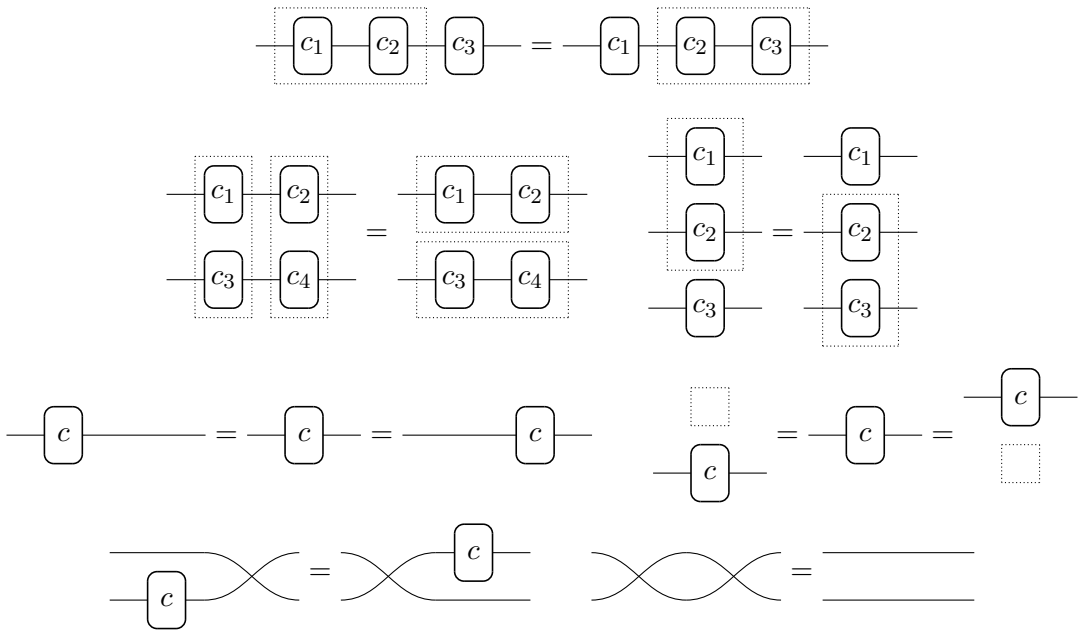

\begin{center}
\[
 \input{tikz/figures/SSM/sequential-associativity.tikz} = \input{tikz/figures/SSM/sequential-associativity-1.tikz}
\]
\[
\input{tikz/figures/SSM/interchange-law.tikz} = \input{tikz/figures/SSM/interchange-law-1.tikz}
 \qquad
 \input{tikz/figures/SSM/parallel-associativity.tikz} = \input{tikz/figures/SSM/parallel-associativity-1.tikz}
\]
\[
\input{tikz/figures/SSM/unit-right2.tikz} = \input{tikz/figures/SSM/c.tikz}
= \input{tikz/figures/SSM/unit-left2.tikz}
\quad\quad
  \input{tikz/figures/SSM/parallel-unit-above.tikz} = \input{tikz/figures/SSM/c.tikz} =  \input{tikz/figures/SSM/parallel-unit-below.tikz}
\]
\[
\input{tikz/figures/SSM/sym-natural.tikz} = \input{tikz/figures/SSM/sym-natural-1.tikz}
\qquad
\input{tikz/figures/SSM/sym-iso.tikz} = \input{tikz/figures/SSM/id2.tikz}
\]
\caption{Laws of Symmetric Strict Monoidal (SSM) Categories.}
\label{fig:lawsSMC}
\end{center}
\end{figure}

In graphical notation, the diagrams are closed under two symmetries: \emph{Mirror-Image} and \emph{Color-Swap}.
The theorems and their corresponding symmetries are summarized in Figures~\ref{fig:Commutative_Comonoid} to~\ref{fig:Minimum_and_maximum}.
Finally, Figure~\ref{fig:typerelations} outlines statements determining whether a relation is classified as total (TOT), deterministic (DET), surjective (SUR), or injective (INJ).

\begin{lemma}
  \label{lem:pair_of_matrices_1}
  Let $\Map{A}$ and $\Map{B}$ be two linear maps, then
  \[\input{tikz/figures/rcd_for_pairs_of_matrices/Lemma1/lemma.tikz}.\]
\end{lemma}
\begin{proof}
  \begin{hcalculation}[\subseteq]{\Map{B}}
    \hstep{\input{tikz/figures/rcd_for_pairs_of_matrices/Lemma1/proof1.tikz}}{Fig.\ref{fig:typerelations}-TOT: \input{tikz/figures/theorems/totalid.tikz}}
    \hstep{\input{tikz/figures/rcd_for_pairs_of_matrices/Lemma1/proof2.tikz}}{Hyp: $\Map{A}\subseteq \Map{B}$}
    \hstep{\Map{A}}{Fig.\ref{fig:typerelations}-DET: \input{tikz/figures/theorems/detid.tikz}}
  \end{hcalculation}
\end{proof}

\begin{lemma}
\label{lem:break_a}
For all maps \Map{A}, there exists a map \Map{A_2} such that
\[ \input{tikz/figures/rcd_for_cospan/lemma_break/lemma_hyp1.tikz} = \input{tikz/figures/rcd_for_cospan/lemma_break/lemma_hyp2.tikz}.\]
\end{lemma}
\begin{proof}
\begin{hcalculation}[=]{\input{tikz/figures/rcd_for_cospan/lemma_break/lemma_hyp1.tikz}}
    \hstep{\input{tikz/figures/rcd_for_cospan/lemma_break/step1.tikz}}{Fig.\ref{fig:Split}: \input{tikz/figures/theorems/split.tikz}}
    \hstep{\input{tikz/figures/rcd_for_cospan/lemma_break/step2.tikz}}{Fig.\ref{fig:typerelations}-DET: \input{tikz/figures/theorems/deteq.tikz}}
    \hstep{\input{tikz/figures/rcd_for_cospan/lemma_break/lemma_hyp2.tikz}}{Fig.\ref{fig:Commutative_Comonoid}: \input{tikz/figures/theorems/unit.tikz}}
  \end{hcalculation}
\end{proof}

\begin{lemma}
\label{thm:subspace_image}
    For every subspace $\input{tikz/figures/rcd_for_cospan/cospan/subspace.tikz}$, there exist linear maps $\Map{X_1}$ and $\Map{X_2}$ such that
\begin{enumerate}
    \item[(i)] \input{tikz/figures/rcd_for_cospan/cospan/subspace-image.tikz},
    \item[(ii)] \input{tikz/figures/rcd_for_cospan/cospan/subspace-kernel.tikz}.
  \end{enumerate}
\end{lemma}
\begin{proof}
  \begin{itemize}
  \item[$(i)$] \citep[Proposition~2.7]{axler2024linear}.
  \item[$(ii)$] Analogous by symmetry.
  \end{itemize}
\end{proof}

\begin{lemma}[Cospan and Span]\label{lem:relation_in_cospan_form_complete} For all relations $\Rel{R}$,
  \begin{enumerate}
    \item[(i)] $\exists \Map{A}, \exists \Map{B}$, such that $\Rel{R} = \CoSpan{A}{B} $,
    \item[(ii)] $\exists \Map{C}, \exists \Map{D}$, such that $\Rel{R} = \Span{C}{D} $.
  \end{enumerate}
\end{lemma}
\begin{proof}
We prove item~(i). The other is analogous.
  \begin{hcalculation}[=]{\Rel{R}}
    \hstep{\input{tikz/figures/rcd_for_cospan/cospan/lemmacospan/step1.tikz}}{Fig.\ref{fig:Snake}: \input{tikz/figures/theorems/snake.tikz}}
    \hstep{\input{tikz/figures/rcd_for_cospan/cospan/lemmacospan/step2.tikz}}{Lem.\ref{thm:subspace_image}:\input{tikz/figures/theorems/subspace-kernel.tikz}}
    \hstep{\input{tikz/figures/rcd_for_cospan/cospan/lemmacospan/step3.tikz}}{Fig.\ref{fig:Split}: \input{tikz/figures/theorems/split.tikz}}
    \hstep{\input{tikz/figures/rcd_for_cospan/cospan/lemmacospan/step3-1.tikz}}{Fig.\ref{fig:typerelations}-DET: \input{tikz/figures/theorems/whiteball.tikz}}
    \hstep{\input{tikz/figures/rcd_for_cospan/cospan/lemmacospan/step4.tikz}}{Fig.\ref{fig:Flip}: \input{tikz/figures/theorems/flip.tikz}}
    \hstep{\input{tikz/figures/rcd_for_cospan/cospan/lemmacospan/step5.tikz}}{Fig.\ref{fig:Snake}: \input{tikz/figures/theorems/snake.tikz}}
  \end{hcalculation}
\end{proof}

\begin{lemma}\label{lemma:abstract}
  For all linear maps $\Map{A}, \Map{B}, \Map{C}, \Map{D}$,
  \[\CoSpan{A}{B}=\CoSpan{C}{D} \iff \input{tikz/figures/rcd_for_pairs_of_matrices/TheoremAbstract/Lemma/lemma.tikz}\]
\end{lemma}
\begin{proof}
  \begin{hcalculation}[\iff]{\input{tikz/figures/rcd_for_pairs_of_matrices/TheoremAbstract/Lemma/lemma.tikz}}
    \hstep{\input{tikz/figures/rcd_for_pairs_of_matrices/TheoremAbstract/Lemma/step7.tikz}}{Fig.\ref{fig:typerelations}-DET: \input{tikz/figures/theorems/whiteball.tikz}}
    \hstep{\input{tikz/figures/rcd_for_pairs_of_matrices/TheoremAbstract/Lemma/step8.tikz}}{Fig.\ref{fig:Flip}: \input{tikz/figures/theorems/flip.tikz}}
    \hstep{\CoSpan{A}{B}=\CoSpan{C}{D}}{Fig.\ref{fig:Compare}: \input{tikz/figures/theorems/comparison.tikz}}
  \end{hcalculation}
\end{proof}

\begin{figure}[htb]
    \centering
    \resizebox{\textwidth}{!}{\input{tikz/figures/Appendix/Commutative_Comonoid.tikz}}
    \caption{Commutative Comonoid}
    \label{fig:Commutative_Comonoid}
\end{figure}

\begin{figure}[htb]
    \centering
    \resizebox{\textwidth}{!}{\input{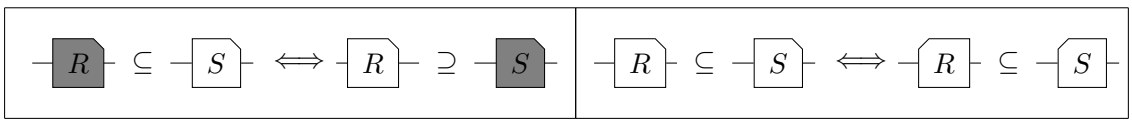}}
    \caption{The Symmetry Theorems}
    \label{fig:Symmetry}
\end{figure}

\begin{figure}[htb]
    \centering
    \input{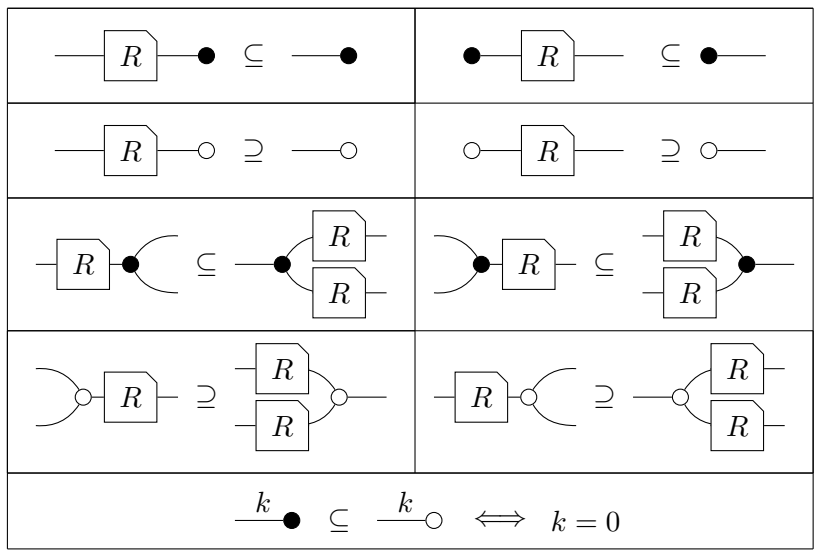}
    \caption{The Inequality Theorems}
    \label{fig:Inequality}
\end{figure}

\begin{figure}[htb]
    \centering
    \resizebox{\textwidth}{!}{\input{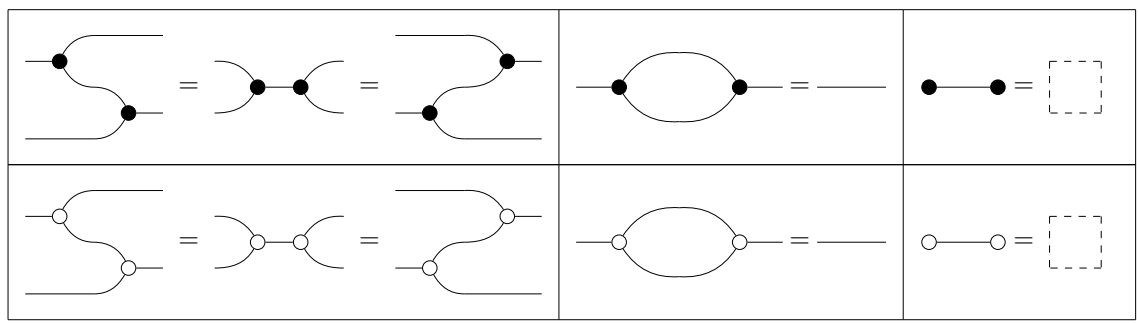}}
    \caption{Frobenius Algebra}
    \label{fig:Frobenius_Algebra}
\end{figure}

\begin{figure}[htb]
    \centering
    \input{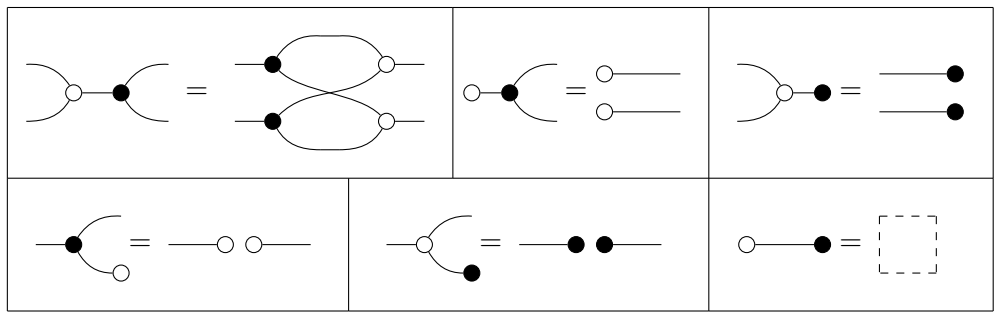}
    \caption{Bialgebra}
    \label{fig:Bialgebra}
\end{figure}

\begin{figure}[htb]
    \centering
    \resizebox{\textwidth}{!}{\input{tikz/figures/Appendix/Flip.tikz}}
    \caption{Flip}
    \label{fig:Flip}
\end{figure}

\begin{figure}[htb]
    \centering
    \input{tikz/figures/Appendix/Snake.tikz}
    \caption{Snake}
    \label{fig:Snake}
\end{figure}

\begin{figure}[htb]
    \centering
    \begin{minipage}[t]{.6\textwidth}
      \centering
      \input{tikz/figures/Appendix/Connect_sum.tikz}
      \captionof{figure}{Connect sum}
      \label{fig:Connect_sum}
    \end{minipage}%
    \hfill\begin{minipage}[t]{.4\textwidth}
      \centering
      \input{tikz/figures/Appendix/Split.tikz}
      \captionof{figure}{Split}
      \label{fig:Split}
    \end{minipage}\\
    \begin{minipage}[t]{.6\textwidth}
      \centering
      \input{tikz/figures/Appendix/Compare.tikz}
      \captionof{figure}{Compare}
      \label{fig:Compare}
    \end{minipage}%
    \hfill\begin{minipage}[t]{.4\textwidth}
      \centering
      \input{tikz/figures/Appendix/Minimum_and_maximum.tikz}
      \captionof{figure}{Minimum and maximum relations}
      \label{fig:Minimum_and_maximum}
    \end{minipage}\\
\end{figure}

\begin{figure}[htb]
    \centering
    \resizebox{\textwidth}{!}{\input{tikz/figures/Appendix/typesofrelations.tikz}}
    \caption{Each column of the table is a series of equivalent statements.
      Furthermore, the opposite inequalities for the second and third row hold for any relation (See Figure~\ref{fig:Inequality}).
      We are, thus, allowed to treat these characterizations as equalities.
    }
    \label{fig:typerelations}
\end{figure}

\FloatBarrier

\bibliographystyle{fundam}
\bibliography{bibliography}

\end{document}

%% file: tikz/gla.tikzstyles
\tikzstyle{Black}=[fill=black, draw=black, shape=circle, inner sep=0, minimum size=6pt]
\tikzstyle{White}=[fill=white, draw=black, shape=circle, inner sep=0, minimum size=6pt]
\tikzstyle{Red}=[fill=red, draw=black, shape=circle]
\tikzstyle{Rel}=[fill=white, draw=black, shape=chamfered rectangle, minimum size=8, chamfered rectangle corners=north east, chamfered rectangle xsep=2, chamfered rectangle ysep=2, tikzit shape=rectangle]
\tikzstyle{CoRel}=[fill=white, draw=black, shape=chamfered rectangle, minimum size=8, chamfered rectangle corners=north west, chamfered rectangle xsep=2, chamfered rectangle ysep=2, tikzit shape=rectangle, tikzit draw=red]
\tikzstyle{RelBig}=[fill=white, draw=black, shape=chamfered rectangle, minimum size=24, chamfered rectangle corners=north east, chamfered rectangle xsep=4, chamfered rectangle ysep=4, tikzit shape=rectangle]
\tikzstyle{CoRelBig}=[fill=white, draw=black, shape=chamfered rectangle, minimum size=24, chamfered rectangle corners=north west, chamfered rectangle xsep=4, chamfered rectangle ysep=4, tikzit shape=rectangle, tikzit draw=red]
\tikzstyle{Map}=[fill=white, draw=black, shape=rounded rectangle, minimum width=23, minimum height=16, rounded rectangle west arc=none, tikzit shape=rectangle]
\tikzstyle{MapSmall}=[fill=white, draw=black, shape=rounded rectangle, minimum width=12, minimum height=8, rounded rectangle west arc=none, tikzit shape=rectangle]
\tikzstyle{CoMap}=[fill=white, draw=black, shape=rounded rectangle, minimum width=23, minimum height=16, rounded rectangle east arc=none, tikzit shape=rectangle, tikzit draw=red]
\tikzstyle{MapBig}=[fill=white, draw=black, shape=rounded rectangle, minimum width=24, minimum height=32, rounded rectangle west arc=none, tikzit shape=rectangle]
\tikzstyle{CoMapBig}=[fill=white, draw=black, shape=rounded rectangle, minimum width=24, minimum height=32, rounded rectangle east arc=none, tikzit shape=rectangle, tikzit draw=red]
\tikzstyle{RelGray}=[fill=gray, draw=black, shape=chamfered rectangle, minimum size=8, chamfered rectangle corners=north east, chamfered rectangle xsep=2, chamfered rectangle ysep=2, tikzit shape=rectangle, tikzit fill={rgb,255: red,191; green,191; blue,191}]
\tikzstyle{CoRelGray}=[fill=gray, draw=black, shape=chamfered rectangle, tikzit draw=red, tikzit fill={rgb,255: red,191; green,191; blue,191}, minimum size=8, chamfered rectangle corners=north west, chamfered rectangle xsep=2, chamfered rectangle ysep=2, tikzit shape=rectangle]
\tikzstyle{Empty}=[fill=white, draw=black, shape=rectangle, dashed, minimum size=16, tikzit draw=black]
\tikzstyle{subspace}=[fill=white, draw=black, shape=rectangle]
\tikzstyle{Inv}=[fill={blue!20}, draw=blue, shape=rounded rectangle, minimum width=23, minimum height=16, rounded rectangle west arc=none, tikzit shape=rectangle]
\tikzstyle{CoInv}=[fill={blue!20}, draw=blue, shape=rounded rectangle, minimum width=23, minimum height=16, rounded rectangle east arc=none, tikzit shape=rectangle, tikzit draw=red]
\tikzstyle{InvBig}=[fill={blue!20}, draw=blue, shape=rounded rectangle, minimum width=24, minimum height=32, rounded rectangle west arc=none, tikzit shape=rectangle]
\tikzstyle{CoInvBig}=[fill={blue!20}, draw=blue, shape=rounded rectangle, minimum width=24, minimum height=32, rounded rectangle east arc=none, tikzit shape=rectangle, tikzit draw=red]

\tikzstyle{new edge style 3}=[-, dashed]

%% file: tikz/figures/introduction_figures/0_by_n_relation.tikz
\begin{tikzpicture}
	\begin{pgfonlayer}{nodelayer}
		\node [style=Rel] (0) at (0, 0) {$R$};
		\node [style=none] (1) at (0.5, 0) {};
		\node [style=none] (5) at (1.75, 0) {};
	\end{pgfonlayer}
	\begin{pgfonlayer}{edgelayer}
		\draw (1.center) to (5.center);
	\end{pgfonlayer}
\end{tikzpicture}

%% file: tikz/figures/introduction_figures/image_of_A.tikz
\begin{tikzpicture}
	\begin{pgfonlayer}{nodelayer}
		\node [style=Rel] (3) at (-0.75, 0) {$\im(A)$};
		\node [style=none] (4) at (2.25, 0) {};
	\end{pgfonlayer}
	\begin{pgfonlayer}{edgelayer}
		\draw (4.center) to (3);
	\end{pgfonlayer}
\end{tikzpicture}

%% file: tikz/figures/introduction_figures/kernel_of_A.tikz
\begin{tikzpicture}
	\begin{pgfonlayer}{nodelayer}
		\node [style=Rel] (3) at (-0.75, 0) {$\ker(A)$};
		\node [style=none] (4) at (2.25, 0) {};
	\end{pgfonlayer}
	\begin{pgfonlayer}{edgelayer}
		\draw (4.center) to (3);
	\end{pgfonlayer}
\end{tikzpicture}

%% file: tikz/figures/rcd_for_cospan/cospan/canonicaldecomposition/step1.tikz
\begin{tikzpicture}
	\begin{pgfonlayer}{nodelayer}
		\node [style=none] (0) at (-8, 0) {};
		\node [style=Inv] (1) at (-6.5, 0) {$P_{\mbox{er}}$};
		\node [style=none] (2) at (-6.5, 0.25) {};
		\node [style=none] (3) at (-6.5, -0.25) {};
		\node [style=Black] (4) at (-5, 0.25) {};
		\node [style=Black] (5) at (-5, -0.25) {};
		\node [style=White] (6) at (0, 0.25) {};
		\node [style=White] (7) at (0, -0.25) {};
		\node [style=Map] (8) at (-2.5, 0.5) {$M$};
		\node [style=Black] (9) at (-3, -0.5) {};
		\node [style=White] (10) at (-2.25, -0.5) {};
		\node [style=Black] (11) at (-3, -1.25) {};
		\node [style=White] (12) at (-2.25, -1.25) {};
		\node [style=Inv] (13) at (1.5, 0) {$Q^{-1}$};
		\node [style=none] (14) at (1.5, 0.25) {};
		\node [style=none] (15) at (1.5, -0.25) {};
		\node [style=none] (16) at (3, 0) {};
	\end{pgfonlayer}
	\begin{pgfonlayer}{edgelayer}
		\draw (2.center) to (4);
		\draw (3.center) to (5);
		\draw [bend left=15] (5) to (8);
		\draw [bend right=15] (8) to (6);
		\draw [bend left=45] (4) to (6);
		\draw [bend right=15] (4) to (9);
		\draw [bend left=345] (7) to (10);
		\draw [bend right, looseness=0.75] (5) to (11);
		\draw [bend right, looseness=0.75] (12) to (7);
		\draw (6) to (14.center);
		\draw (7) to (15.center);
		\draw (13) to (16.center);
		\draw (0.center) to (1);
	\end{pgfonlayer}
\end{tikzpicture}

%% file: tikz/figures/rcd_for_cospan/cospan/canonicaldecomposition/step2.tikz
\begin{tikzpicture}
	\begin{pgfonlayer}{nodelayer}
		\node [style=none] (0) at (-8.25, 0) {};
		\node [style=Inv] (1) at (-6.5, 0) {$P_{\mbox{er}}$};
		\node [style=none] (2) at (-6.5, 0.25) {};
		\node [style=none] (3) at (-6.5, -0.25) {};
		\node [style=none] (4) at (-6, 0.25) {};
		\node [style=Black] (5) at (-4.75, -0.25) {};
		\node [style=White] (6) at (-1, 0.25) {};
		\node [style=White] (7) at (0.5, -0.75) {};
		\node [style=Map] (8) at (-2.5, 0) {$M$};
		\node [style=Black] (11) at (0, -0.75) {};
		\node [style=Inv] (13) at (2, 0) {$Q^{-1}$};
		\node [style=none] (14) at (2, 0.25) {};
		\node [style=none] (15) at (1.75, 0) {};
		\node [style=none] (16) at (3.5, 0) {};
		\node [style=none] (18) at (-2, 1.5) {};
		\node [style=none] (19) at (-2, -1.5) {};
		\node [style=none] (20) at (-7.75, 1.5) {};
		\node [style=none] (21) at (-7.75, -1.5) {};
	\end{pgfonlayer}
	\begin{pgfonlayer}{edgelayer}
		\draw (2.center) to (4.center);
		\draw (3.center) to (5);
		\draw [bend left=15] (5) to (8);
		\draw [bend right=15] (8) to (6);
		\draw [bend left] (4.center) to (6);
		\draw [in=-180, out=-30, looseness=0.75] (5) to (11);
		\draw (6) to (14.center);
		\draw [bend right=15] (7) to (15.center);
		\draw (13) to (16.center);
		\draw [style=new edge style 3] (20.center) to (21.center);
		\draw [style=new edge style 3] (20.center) to (18.center);
		\draw [style=new edge style 3, bend left=75, looseness=1.75] (18.center) to (19.center);
		\draw [style=new edge style 3] (21.center) to (19.center);
		\draw (0.center) to (1);
	\end{pgfonlayer}
\end{tikzpicture}

%% file: tikz/figures/rcd_for_cospan/cospan/canonicaldecomposition/step3.tikz
\begin{tikzpicture}
	\begin{pgfonlayer}{nodelayer}
		\node [style=Inv] (13) at (2.5, 0) {$Q^{-1}$};
		\node [style=none] (16) at (4.25, 0) {};
		\node [style=Inv] (17) at (-1.25, 0) {$P^{-1}$};
		\node [style=none] (18) at (-3, 0) {};
		\node [style=none] (19) at (-1.25, 0.25) {};
		\node [style=none] (20) at (2.5, 0.25) {};
		\node [style=none] (21) at (-1.25, -0.25) {};
		\node [style=none] (22) at (2.5, -0.25) {};
		\node [style=Black] (23) at (0.25, -0.25) {};
		\node [style=White] (24) at (1, -0.25) {};
		\node [style=none] (25) at (0.5, 0.5) {$\scriptstyle{r}$};
		\node [style=none] (26) at (-2.5, 0.25) {$\scriptstyle{m}$};
		\node [style=none] (27) at (3.75, 0.25) {$\scriptstyle{n}$};
	\end{pgfonlayer}
	\begin{pgfonlayer}{edgelayer}
		\draw (13) to (16.center);
		\draw (18.center) to (17);
		\draw (19.center) to (20.center);
		\draw (21.center) to (23);
		\draw (24) to (22.center);
	\end{pgfonlayer}
\end{tikzpicture}

%% file: tikz/figures/rcd_for_cospan/cospan/canonicaldecomposition/defP.tikz
\begin{tikzpicture}
	\begin{pgfonlayer}{nodelayer}
		\node [style=none] (15) at (0, 0) {$:=$};
		\node [style=none] (16) at (-0.75, 0.25) {};
		\node [style=Inv] (17) at (-2.25, 0) {$P^{-1}$};
		\node [style=none] (18) at (-3.75, 0) {};
		\node [style=none] (20) at (0.75, 0) {};
		\node [style=Inv] (21) at (2.25, 0) {$P_{\mbox{er}}$};
		\node [style=none] (22) at (2.25, 0.25) {};
		\node [style=none] (23) at (2.25, -0.25) {};
		\node [style=none] (24) at (2.75, 0.25) {};
		\node [style=Black] (25) at (3.75, -0.25) {};
		\node [style=White] (26) at (6.75, 0.25) {};
		\node [style=Map] (28) at (5.25, 0) {$M$};
		\node [style=none] (29) at (5.25, -0.75) {};
		\node [style=none] (31) at (8, 0.25) {};
		\node [style=none] (32) at (-2.25, 0.25) {};
		\node [style=none] (33) at (-2.25, -0.25) {};
		\node [style=none] (34) at (-0.75, -0.25) {};
		\node [style=none] (35) at (8, -0.75) {};
	\end{pgfonlayer}
	\begin{pgfonlayer}{edgelayer}
		\draw (18.center) to (17);
		\draw (22.center) to (24.center);
		\draw (23.center) to (25);
		\draw [bend left=15] (25) to (28);
		\draw [bend right=15] (28) to (26);
		\draw [bend left] (24.center) to (26);
		\draw [bend right=15, looseness=1.25] (25) to (29.center);
		\draw (26) to (31.center);
		\draw (20.center) to (21);
		\draw (32.center) to (16.center);
		\draw (33.center) to (34.center);
		\draw (29.center) to (35.center);
	\end{pgfonlayer}
\end{tikzpicture}

%% file: tikz/figures/rcd_for_cospan/cospan/canonicaldecomposition/step4.tikz
\begin{tikzpicture}
	\begin{pgfonlayer}{nodelayer}
		\node [style=CoInv] (13) at (2.75, 0) {$Q$};
		\node [style=none] (16) at (4.75, 0) {};
		\node [style=CoInv] (17) at (-1.5, 0) {$P$};
		\node [style=none] (18) at (-3.75, 0) {};
		\node [style=none] (19) at (-1.5, 0.25) {};
		\node [style=none] (20) at (2.75, 0.25) {};
		\node [style=none] (21) at (-1.5, -0.25) {};
		\node [style=none] (22) at (2.75, -0.25) {};
		\node [style=Black] (23) at (0.25, -0.25) {};
		\node [style=White] (24) at (1, -0.25) {};
		\node [style=none] (25) at (0.5, 0.5) {$\scriptstyle{r}$};
		\node [style=none] (26) at (-3.25, 0.25) {$\scriptstyle{m}$};
		\node [style=none] (27) at (4.25, 0.25) {$\scriptstyle{n}$};
	\end{pgfonlayer}
	\begin{pgfonlayer}{edgelayer}
		\draw (13) to (16.center);
		\draw (18.center) to (17);
		\draw (19.center) to (20.center);
		\draw (21.center) to (23);
		\draw (24) to (22.center);
	\end{pgfonlayer}
\end{tikzpicture}

%% file: tikz/figures/no_lolipop/step2.tikz
\begin{tikzpicture}
	\begin{pgfonlayer}{nodelayer}
		\node [style=none] (0) at (-2.5, 0) {};
		\node [style=Black] (1) at (-1, 0) {};
		\node [style=none] (7) at (1, 0) {};
		\node [style=Inv] (8) at (2.5, 0) {$P$};
		\node [style=CoInv] (9) at (7.75, 0) {$Q$};
		\node [style=Black] (10) at (9.75, 0) {};
		\node [style=none] (11) at (2.75, 0) {};
		\node [style=none] (12) at (7.5, 0) {};
		\node [style=none] (13) at (2.5, 0) {};
		\node [style=none] (14) at (7.75, 0) {};
		\node [style=Black] (15) at (4.75, 1) {};
		\node [style=White] (16) at (5.5, 1) {};
		\node [style=none] (17) at (1.5, 0.25) {$\scriptstyle{m}$};
		\node [style=none] (18) at (9, 0.25) {$\scriptstyle{n}$};
		\node [style=none] (19) at (5.25, 0.25) {$\scriptstyle{r}$};
		\node [style=none] (23) at (4, 1.25) {$\scriptstyle{k_I}$};
		\node [style=none] (24) at (6.25, 1.25) {$\scriptstyle{k_S}$};
		\node [style=none] (25) at (0, 0) {$\subseteq$};
		\node [style=none] (26) at (2.5, 0) {};
		\node [style=none] (27) at (7.75, 0) {};
		\node [style=White] (28) at (4.75, -1) {};
		\node [style=Black] (29) at (5.5, -1) {};
		\node [style=none] (30) at (4, -0.75) {$\scriptstyle{k_T}$};
		\node [style=none] (31) at (6.25, -0.75) {$\scriptstyle{k_D}$};
	\end{pgfonlayer}
	\begin{pgfonlayer}{edgelayer}
		\draw (0.center) to (1);
		\draw (11.center) to (12.center);
		\draw (10) to (9);
		\draw (8) to (7.center);
		\draw [bend right] (15) to (13.center);
		\draw [bend left] (16) to (14.center);
		\draw [bend left] (28) to (26.center);
		\draw [bend right] (29) to (27.center);
	\end{pgfonlayer}
\end{tikzpicture}

%% file: tikz/figures/no_lolipop/step3.tikz
\begin{tikzpicture}
	\begin{pgfonlayer}{nodelayer}
		\node [style=none] (0) at (-4.5, 0) {};
		\node [style=Black] (1) at (-3, 0) {};
		\node [style=none] (2) at (-1.25, 0) {};
		\node [style=Inv] (3) at (0.25, 0) {$P$};
		\node [style=none] (6) at (0.5, 0) {};
		\node [style=Black] (7) at (4.25, 0) {};
		\node [style=none] (8) at (0.25, 0.25) {};
		\node [style=Black] (9) at (4.25, 0.75) {};
		\node [style=Black] (10) at (2, 0.75) {};
		\node [style=White] (11) at (2.75, 0.75) {};
		\node [style=none] (12) at (-0.75, 0.25) {$\scriptstyle{m}$};
		\node [style=none] (14) at (2.5, 0.25) {$\scriptstyle{r}$};
		\node [style=none] (15) at (1.25, 1) {$\scriptstyle{k_I}$};
		\node [style=none] (16) at (3.5, 1) {$\scriptstyle{k_S}$};
		\node [style=none] (17) at (-2, 0) {$\subseteq$};
		\node [style=none] (18) at (0.25, -0.25) {};
		\node [style=Black] (19) at (4.25, -0.75) {};
		\node [style=White] (20) at (2, -0.75) {};
		\node [style=Black] (21) at (2.75, -0.75) {};
		\node [style=none] (22) at (1.25, -1) {$\scriptstyle{k_T}$};
		\node [style=none] (23) at (3.5, -0.5) {$\scriptstyle{k_D}$};
	\end{pgfonlayer}
	\begin{pgfonlayer}{edgelayer}
		\draw (0.center) to (1);
		\draw (6.center) to (7);
		\draw (3) to (2.center);
		\draw [bend right=15] (10) to (8.center);
		\draw (11) to (9);
		\draw [bend left=15] (20) to (18.center);
		\draw (21) to (19);
	\end{pgfonlayer}
\end{tikzpicture}

%% file: tikz/figures/no_lolipop/step4.tikz
\begin{tikzpicture}
	\begin{pgfonlayer}{nodelayer}
		\node [style=none] (0) at (-4.5, 0) {};
		\node [style=Black] (1) at (-3, 0) {};
		\node [style=none] (2) at (-1.25, 0) {};
		\node [style=Inv] (3) at (0.25, 0) {$P$};
		\node [style=none] (6) at (0.5, 0) {};
		\node [style=Black] (7) at (2.5, 0) {};
		\node [style=none] (8) at (0.25, 0) {};
		\node [style=Black] (10) at (2.5, 1) {};
		\node [style=none] (12) at (-0.75, 0.25) {$\scriptstyle{m}$};
		\node [style=none] (14) at (1.75, 0.25) {$\scriptstyle{r}$};
		\node [style=none] (15) at (1.75, 1.25) {$\scriptstyle{k_I}$};
		\node [style=none] (17) at (-2, 0) {$\subseteq$};
		\node [style=none] (18) at (0.25, 0) {};
		\node [style=White] (20) at (2.5, -1) {};
		\node [style=none] (22) at (1.75, -0.75) {$\scriptstyle{k_T}$};
	\end{pgfonlayer}
	\begin{pgfonlayer}{edgelayer}
		\draw (0.center) to (1);
		\draw (6.center) to (7);
		\draw (3) to (2.center);
		\draw [bend left=330] (10) to (8.center);
		\draw [bend left] (20) to (18.center);
	\end{pgfonlayer}
\end{tikzpicture}

%% file: tikz/figures/no_lolipop/step5.tikz
\begin{tikzpicture}
	\begin{pgfonlayer}{nodelayer}
		\node [style=none] (0) at (-3.5, 0) {};
		\node [style=Black] (1) at (-2, 0) {};
		\node [style=CoInv] (3) at (-3.5, 0) {$P$};
		\node [style=none] (6) at (0, 0) {};
		\node [style=Black] (7) at (1.75, 0) {};
		\node [style=none] (8) at (0, 1) {};
		\node [style=Black] (10) at (1.75, 1) {};
		\node [style=none] (14) at (1, 0.25) {$\scriptstyle{r}$};
		\node [style=none] (15) at (1, 1.25) {$\scriptstyle{k_I}$};
		\node [style=none] (17) at (-1, 0) {$\subseteq$};
		\node [style=none] (18) at (0, -1) {};
		\node [style=White] (20) at (1.75, -1) {};
		\node [style=none] (22) at (1, -0.75) {$\scriptstyle{k_T}$};
		\node [style=none] (23) at (-5.75, 0) {};
		\node [style=none] (24) at (-4, 0) {};
		\node [style=none] (25) at (-5.75, 1) {};
		\node [style=none] (26) at (-3.5, 0) {};
		\node [style=none] (27) at (-5, 0.25) {$\scriptstyle{r}$};
		\node [style=none] (28) at (-5, 1.25) {$\scriptstyle{k_I}$};
		\node [style=none] (29) at (-5.75, -1) {};
		\node [style=none] (30) at (-3.5, 0) {};
		\node [style=none] (31) at (-5, -0.75) {$\scriptstyle{k_T}$};
	\end{pgfonlayer}
	\begin{pgfonlayer}{edgelayer}
		\draw (0.center) to (1);
		\draw (6.center) to (7);
		\draw (10) to (8.center);
		\draw (20) to (18.center);
		\draw (23.center) to (24.center);
		\draw [bend right] (26.center) to (25.center);
		\draw [bend left] (30.center) to (29.center);
	\end{pgfonlayer}
\end{tikzpicture}

%% file: tikz/figures/no_lolipop/step6.tikz
\begin{tikzpicture}
	\begin{pgfonlayer}{nodelayer}
		\node [style=none] (6) at (1.25, 0) {};
		\node [style=Black] (7) at (3.5, 0) {};
		\node [style=none] (8) at (1.25, 1) {};
		\node [style=Black] (10) at (3.5, 1) {};
		\node [style=none] (14) at (2.75, 0.25) {$\scriptstyle{r}$};
		\node [style=none] (15) at (2.75, 1.25) {$\scriptstyle{k_I}$};
		\node [style=none] (17) at (0, 0) {$\subseteq$};
		\node [style=none] (18) at (1.25, -1) {};
		\node [style=White] (20) at (3.5, -1) {};
		\node [style=none] (22) at (2.75, -0.75) {$\scriptstyle{k_T}$};
		\node [style=none] (23) at (-3.75, 0) {};
		\node [style=Black] (24) at (-1.5, 0) {};
		\node [style=none] (25) at (-3.75, 1) {};
		\node [style=Black] (26) at (-1.5, 1) {};
		\node [style=none] (27) at (-2.75, 0.25) {$\scriptstyle{r}$};
		\node [style=none] (28) at (-2.75, 1.25) {$\scriptstyle{k_I}$};
		\node [style=none] (29) at (-3.75, -1) {};
		\node [style=Black] (30) at (-1.5, -1) {};
		\node [style=none] (31) at (-2.75, -0.75) {$\scriptstyle{k_T}$};
	\end{pgfonlayer}
	\begin{pgfonlayer}{edgelayer}
		\draw (6.center) to (7);
		\draw (10) to (8.center);
		\draw (20) to (18.center);
		\draw (23.center) to (24);
		\draw (26) to (25.center);
		\draw (30) to (29.center);
	\end{pgfonlayer}
\end{tikzpicture}

%% file: tikz/figures/poset_of_fundamental_properties/step3.tikz
\begin{tikzpicture}
	\begin{pgfonlayer}{nodelayer}
		\node [style=none] (0) at (-4.75, 0) {};
		\node [style=Inv] (1) at (-2.75, 0) {$P$};
		\node [style=CoInv] (2) at (1.5, 0) {$Q$};
		\node [style=none] (3) at (3.5, 0) {};
		\node [style=none] (4) at (-2.5, -0.25) {};
		\node [style=none] (5) at (1.25, -0.25) {};
		\node [style=none] (6) at (-2.5, 0.25) {};
		\node [style=Black] (8) at (-1, 0.25) {};
		\node [style=none] (10) at (-4, 0.25) {$\scriptstyle{m}$};
		\node [style=none] (11) at (2.75, 0.25) {$\scriptstyle{n}$};
		\node [style=none] (12) at (0, 0) {$\scriptstyle{r}$};
		\node [style=none] (13) at (-1.75, 0.5) {$\scriptstyle{k_I}$};
	\end{pgfonlayer}
	\begin{pgfonlayer}{edgelayer}
		\draw (4.center) to (5.center);
		\draw (3.center) to (2);
		\draw (1) to (0.center);
		\draw (8) to (6.center);
	\end{pgfonlayer}
\end{tikzpicture}

%% file: tikz/figures/poset_of_fundamental_properties/step8.tikz
\begin{tikzpicture}
	\begin{pgfonlayer}{nodelayer}
		\node [style=none] (0) at (-4.25, 0) {};
		\node [style=Inv] (1) at (-2.25, 0) {$P$};
		\node [style=CoInv] (2) at (1.5, 0) {$Q$};
		\node [style=none] (3) at (3.5, 0) {};
		\node [style=none] (4) at (-2, 0) {};
		\node [style=none] (5) at (1.25, 0) {};
		\node [style=White] (6) at (-1, 1) {};
		\node [style=none] (8) at (1.5, 0) {};
		\node [style=none] (10) at (-3.5, 0.25) {$\scriptstyle{m}$};
		\node [style=none] (11) at (2.75, 0.25) {$\scriptstyle{n}$};
		\node [style=none] (12) at (-0.25, 0.25) {$\scriptstyle{r}$};
		\node [style=none] (13) at (-0.25, 1.25) {$\scriptstyle{k_I}$};
		\node [style=Black] (14) at (-1, -1) {};
		\node [style=none] (15) at (1.5, 0) {};
		\node [style=none] (16) at (-0.25, -0.75) {$\scriptstyle{k_T}$};
	\end{pgfonlayer}
	\begin{pgfonlayer}{edgelayer}
		\draw (4.center) to (5.center);
		\draw (3.center) to (2);
		\draw (1) to (0.center);
		\draw [bend right] (8.center) to (6);
		\draw [bend right] (14) to (15.center);
	\end{pgfonlayer}
\end{tikzpicture}

%% file: tikz/figures/thm_surjective_injective_matrix/stat_1.tikz
\begin{tikzpicture}
	\begin{pgfonlayer}{nodelayer}
		\node [style=none] (0) at (-6.25, 5) {};
		\node [style=Map] (1) at (-5, 5) {$A$};
		\node [style=none] (3) at (-3.75, 5) {};
		\node [style=none] (4) at (-2.5, 5) {is SUR};
		\node [style=none] (5) at (1, 5) {};
		\node [style=Map] (6) at (2.25, 5) {$A$};
		\node [style=none] (8) at (3.5, 5) {};
		\node [style=none] (9) at (4.75, 5) {is INJ};
		\node [style=none] (11) at (-7.25, 4) {};
		\node [style=none] (12) at (0, 6) {};
		\node [style=none] (13) at (0, 4) {};
		\node [style=none] (14) at (7, 4) {};
		\node [style=none] (15) at (-7.25, 4) {};
		\node [style=none] (16) at (-7.25, 2.5) {};
		\node [style=none] (17) at (0, 4) {};
		\node [style=none] (18) at (0, 2.5) {};
		\node [style=none] (19) at (7, 2.5) {};
		\node [style=none] (20) at (-7.25, 2.5) {};
		\node [style=none] (21) at (-7.25, -0.75) {};
		\node [style=none] (22) at (0, 2.5) {};
		\node [style=none] (23) at (0, -0.75) {};
		\node [style=none] (24) at (7, -0.75) {};
		\node [style=none] (25) at (-7.25, -0.75) {};
		\node [style=none] (26) at (-7.25, -6.5) {};
		\node [style=none] (27) at (0, -0.75) {};
		\node [style=none] (28) at (0, -6.5) {};
		\node [style=none] (29) at (7, -6.5) {};
		\node [style=none] (30) at (-7.25, -6.5) {};
		\node [style=none] (32) at (0, -6.5) {};
		\node [style=none] (33) at (0, -8.5) {};
		\node [style=none] (35) at (-7.25, -4) {};
		\node [style=none] (36) at (0, -8.5) {};
		\node [style=none] (37) at (0, -4) {};
		\node [style=none] (38) at (-3.75, 3.25) {$k_S = 0$};
		\node [style=none] (39) at (3.5, 3.25) {$k_I = 0$};
		\node [style=none] (40) at (-6.25, 0.25) {};
		\node [style=CoMap] (41) at (-5.25, 0.25) {$S$};
		\node [style=none] (42) at (-4.25, 0.25) {};
		\node [style=none] (43) at (-3.75, 0.25) {$\subseteq$};
		\node [style=none] (44) at (-3.25, 0.25) {};
		\node [style=Map] (45) at (-2.25, 0.25) {$A$};
		\node [style=none] (47) at (-1, 0.25) {};
		\node [style=none] (48) at (1, 0.25) {};
		\node [style=CoMap] (49) at (2, 0.25) {$S$};
		\node [style=none] (50) at (3, 0.25) {};
		\node [style=none] (51) at (3.5, 0.25) {$\supseteq$};
		\node [style=none] (52) at (4, 0.25) {};
		\node [style=Map] (53) at (5, 0.25) {$A$};
		\node [style=none] (54) at (6, 0.25) {};
		\node [style=none] (55) at (-3, -5.25) {$=$};
		\node [style=none] (56) at (-7.25, -5.25) {};
		\node [style=CoMap] (57) at (-6, -5.25) {$A$};
		\node [style=Map] (58) at (-4.5, -5.25) {$A$};
		\node [style=none] (59) at (-3.5, -5.25) {};
		\node [style=none] (64) at (4.25, -5.25) {$\subseteq$};
		\node [style=none] (67) at (0.25, -5.25) {};
		\node [style=Map] (68) at (1.25, -5.25) {$A$};
		\node [style=CoMap] (69) at (2.75, -5.25) {$A$};
		\node [style=none] (70) at (3.75, -5.25) {};
		\node [style=Black] (71) at (-5.75, -7.5) {};
		\node [style=Map] (72) at (-4.5, -7.5) {$A$};
		\node [style=none] (73) at (-3.5, -7.5) {};
		\node [style=none] (74) at (-3, -7.5) {$=$};
		\node [style=Black] (75) at (-2.25, -7.5) {};
		\node [style=none] (77) at (-0.75, -7.5) {};
		\node [style=none] (78) at (-2.5, -3) {};
		\node [style=none] (80) at (-0.5, -3) {};
		\node [style=none] (81) at (-3, -3) {$=$};
		\node [style=none] (82) at (-7.25, -3) {};
		\node [style=Map] (83) at (-6, -3) {$S$};
		\node [style=Map] (84) at (-4.5, -3) {$A$};
		\node [style=none] (85) at (-3.5, -3) {};
		\node [style=none] (86) at (4.75, -3) {};
		\node [style=none] (87) at (6.75, -3) {};
		\node [style=none] (88) at (4.25, -3) {$=$};
		\node [style=none] (89) at (0.25, -3) {};
		\node [style=Map] (90) at (1.25, -3) {$A$};
		\node [style=Map] (91) at (2.75, -3) {$S$};
		\node [style=none] (92) at (3.75, -3) {};
		\node [style=none] (93) at (5, -7.5) {};
		\node [style=White] (94) at (6.25, -7.5) {};
		\node [style=none] (95) at (4.25, -7.5) {$=$};
		\node [style=none] (96) at (0.75, -7.5) {};
		\node [style=Map] (97) at (2, -7.5) {$A$};
		\node [style=White] (98) at (3.25, -7.5) {};
		\node [style=none] (99) at (7, -4) {};
		\node [style=none] (100) at (-6.25, 1.75) {$\exists$};
		\node [style=none] (101) at (-2.5, 1.75) {$\mbox{ such that}$ };
		\node [style=none] (102) at (-6, 1.75) {};
		\node [style=MapSmall] (103) at (-5, 1.75) {$\mbox{S}$};
		\node [style=none] (104) at (-4, 1.75) {};
		\node [style=none] (105) at (1, 1.75) {$\exists$};
		\node [style=none] (106) at (4.75, 1.75) {$\mbox{ such that}$ };
		\node [style=none] (107) at (1.25, 1.75) {};
		\node [style=MapSmall] (108) at (2.25, 1.75) {$\mbox{S}$};
		\node [style=none] (109) at (3.25, 1.75) {};
		\node [style=none] (110) at (-6.25, -1.5) {$\exists$};
		\node [style=none] (111) at (-2.5, -1.5) {$\mbox{ such that}$ };
		\node [style=none] (112) at (-6, -1.5) {};
		\node [style=MapSmall] (113) at (-5, -1.5) {$\mbox{S}$};
		\node [style=none] (114) at (-4, -1.5) {};
		\node [style=none] (115) at (1, -1.5) {$\exists$};
		\node [style=none] (116) at (4.75, -1.5) {$\mbox{ such that}$ };
		\node [style=none] (117) at (1.25, -1.5) {};
		\node [style=MapSmall] (118) at (2.25, -1.5) {$\mbox{S}$};
		\node [style=none] (119) at (3.25, -1.5) {};
		\node [style=none] (120) at (-2.5, -5.25) {};
		\node [style=none] (121) at (-0.5, -5.25) {};
		\node [style=none] (122) at (4.75, -5.25) {};
		\node [style=none] (123) at (6.75, -5.25) {};
	\end{pgfonlayer}
	\begin{pgfonlayer}{edgelayer}
		\draw (0.center) to (1);
		\draw (5.center) to (6);
		\draw (13.center) to (12.center);
		\draw (11.center) to (13.center);
		\draw (13.center) to (14.center);
		\draw (18.center) to (17.center);
		\draw (16.center) to (18.center);
		\draw (18.center) to (19.center);
		\draw (23.center) to (22.center);
		\draw (21.center) to (23.center);
		\draw (23.center) to (24.center);
		\draw (28.center) to (27.center);
		\draw (26.center) to (28.center);
		\draw (28.center) to (29.center);
		\draw (33.center) to (32.center);
		\draw (37.center) to (36.center);
		\draw (41) to (42.center);
		\draw (41) to (40.center);
		\draw (44.center) to (45);
		\draw (49) to (50.center);
		\draw (49) to (48.center);
		\draw (53) to (54.center);
		\draw (52.center) to (53);
		\draw (56.center) to (57);
		\draw (57) to (58);
		\draw (58) to (59.center);
		\draw (67.center) to (68);
		\draw (68) to (69);
		\draw (69) to (70.center);
		\draw (72) to (73.center);
		\draw (72) to (71);
		\draw (82.center) to (83);
		\draw (83) to (84);
		\draw (84) to (85.center);
		\draw (89.center) to (90);
		\draw (90) to (91);
		\draw (91) to (92.center);
		\draw (86.center) to (87.center);
		\draw (97) to (98);
		\draw (96.center) to (97);
		\draw (93.center) to (94);
		\draw (35.center) to (37.center);
		\draw (99.center) to (37.center);
		\draw (103) to (104.center);
		\draw (103) to (102.center);
		\draw (108) to (109.center);
		\draw (108) to (107.center);
		\draw (113) to (114.center);
		\draw (113) to (112.center);
		\draw (118) to (119.center);
		\draw (118) to (117.center);
		\draw (1) to (3.center);
		\draw (6) to (8.center);
		\draw (45) to (47.center);
		\draw (78.center) to (80.center);
		\draw (120.center) to (121.center);
		\draw (75) to (77.center);
		\draw (122.center) to (123.center);
	\end{pgfonlayer}
\end{tikzpicture}

%% file: tikz/figures/rcd_for_pairs_of_matrices/TheoremAbstract/item1/hSV_step1.tikz
\begin{tikzpicture}
	\begin{pgfonlayer}{nodelayer}
		\node [style=White] (0) at (-1.5, 0) {};
		\node [style=Rel] (1) at (0, 0) {$H$};
		\node [style=none] (2) at (1.5, 0) {};
	\end{pgfonlayer}
	\begin{pgfonlayer}{edgelayer}
		\draw (0) to (1);
		\draw (1) to (2.center);
	\end{pgfonlayer}
\end{tikzpicture}

%% file: tikz/figures/rcd_for_pairs_of_matrices/TheoremAbstract/item1/hTOT_step1.tikz
\begin{tikzpicture}
	\begin{pgfonlayer}{nodelayer}
		\node [style=none] (0) at (-1.5, 0) {};
		\node [style=Rel] (1) at (0, 0) {$H$};
		\node [style=Black] (2) at (1.5, 0) {};
	\end{pgfonlayer}
	\begin{pgfonlayer}{edgelayer}
		\draw (0.center) to (1);
		\draw (1) to (2);
	\end{pgfonlayer}
\end{tikzpicture}

%% file: tikz/figures/rcd_for_pairs_of_matrices/TheoremAbstract/item1/hTOT_step5.tikz
\begin{tikzpicture}
	\begin{pgfonlayer}{nodelayer}
		\node [style=none] (6) at (-1, 0) {};
		\node [style=Black] (7) at (1, 0) {};
	\end{pgfonlayer}
	\begin{pgfonlayer}{edgelayer}
		\draw (6.center) to (7);
	\end{pgfonlayer}
\end{tikzpicture}

%% file: tikz/figures/rcd_for_pairs_of_matrices/TheoremAbstract/item4/step1.tikz
\begin{tikzpicture}
	\begin{pgfonlayer}{nodelayer}
		\node [style=none] (24) at (-1.25, 0) {};
		\node [style=Rel] (25) at (0, 0) {$H$};
		\node [style=White] (26) at (1.25, 0) {};
	\end{pgfonlayer}
	\begin{pgfonlayer}{edgelayer}
		\draw (24.center) to (25);
		\draw (25) to (26);
	\end{pgfonlayer}
\end{tikzpicture}

%% file: tikz/figures/rcd_for_pairs_of_matrices/TheoremAbstract/item4/step2.tikz
\begin{tikzpicture}
	\begin{pgfonlayer}{nodelayer}
		\node [style=none] (0) at (-3.25, 0) {};
		\node [style=White] (1) at (-2.25, 0) {};
		\node [style=CoMap] (2) at (-0.75, 0.75) {$D_1$};
		\node [style=CoMap] (3) at (-0.75, -0.75) {$D_2$};
		\node [style=Map] (4) at (0.75, 0.75) {$A$};
		\node [style=Map] (5) at (0.75, -0.75) {$B$};
		\node [style=White] (6) at (2.25, 0) {};
		\node [style=White] (7) at (3.25, 0) {};
	\end{pgfonlayer}
	\begin{pgfonlayer}{edgelayer}
		\draw (0.center) to (1);
		\draw [bend left] (1) to (2);
		\draw (2) to (4);
		\draw [bend left] (4) to (6);
		\draw (6) to (7);
		\draw [bend left] (6) to (5);
		\draw (5) to (3);
		\draw [bend right=330] (3) to (1);
	\end{pgfonlayer}
\end{tikzpicture}

%% file: tikz/figures/rcd_for_pairs_of_matrices/TheoremAbstract/item4/step3.tikz
\begin{tikzpicture}
	\begin{pgfonlayer}{nodelayer}
		\node [style=none] (0) at (-3.25, 0) {};
		\node [style=White] (1) at (-2.25, 0) {};
		\node [style=CoMap] (2) at (-0.75, 0.75) {$D_1$};
		\node [style=CoMap] (3) at (-0.75, -0.75) {$D_2$};
		\node [style=Map] (4) at (0.75, 0.75) {$D_1$};
		\node [style=Map] (5) at (0.75, -0.75) {$D_2$};
		\node [style=White] (6) at (2.25, 0) {};
		\node [style=White] (7) at (3.25, 0) {};
	\end{pgfonlayer}
	\begin{pgfonlayer}{edgelayer}
		\draw (0.center) to (1);
		\draw [bend left] (1) to (2);
		\draw (2) to (4);
		\draw [bend left] (4) to (6);
		\draw (6) to (7);
		\draw [bend left] (6) to (5);
		\draw (5) to (3);
		\draw [bend right=330] (3) to (1);
	\end{pgfonlayer}
\end{tikzpicture}

%% file: tikz/figures/rcd_for_pairs_of_matrices/TheoremAbstract/item4/step4.tikz
\begin{tikzpicture}
	\begin{pgfonlayer}{nodelayer}
		\node [style=none] (0) at (-2.75, 0) {};
		\node [style=White] (1) at (-1.75, 0) {};
		\node [style=none] (2) at (0, 1) {};
		\node [style=none] (3) at (0, -1) {};
		\node [style=none] (4) at (0, 1) {};
		\node [style=none] (5) at (0, -1) {};
		\node [style=White] (6) at (1.75, 0) {};
		\node [style=White] (7) at (2.75, 0) {};
	\end{pgfonlayer}
	\begin{pgfonlayer}{edgelayer}
		\draw (0.center) to (1);
		\draw [bend left] (1) to (2.center);
		\draw (2.center) to (4.center);
		\draw [bend left] (4.center) to (6);
		\draw (6) to (7);
		\draw [bend left] (6) to (5.center);
		\draw (5.center) to (3.center);
		\draw [bend right=330] (3.center) to (1);
	\end{pgfonlayer}
\end{tikzpicture}

%% file: tikz/figures/rcd_for_pairs_of_matrices/TheoremAbstract/item5/step1.tikz
\begin{tikzpicture}
	\begin{pgfonlayer}{nodelayer}
		\node [style=Black] (24) at (-1.25, 0) {};
		\node [style=Rel] (25) at (0, 0) {$H$};
		\node [style=none] (26) at (1.25, 0) {};
	\end{pgfonlayer}
	\begin{pgfonlayer}{edgelayer}
		\draw (24) to (25);
		\draw (25) to (26.center);
	\end{pgfonlayer}
\end{tikzpicture}

%% file: tikz/figures/rcd_for_pairs_of_matrices/TheoremAbstract/item5/step2.tikz
\begin{tikzpicture}
	\begin{pgfonlayer}{nodelayer}
		\node [style=Black] (0) at (-3.25, 0) {};
		\node [style=White] (1) at (-2.25, 0) {};
		\node [style=CoMap] (2) at (-0.75, 0.75) {$D_1$};
		\node [style=CoMap] (3) at (-0.75, -0.75) {$D_2$};
		\node [style=Map] (4) at (0.75, 0.75) {$A$};
		\node [style=Map] (5) at (0.75, -0.75) {$B$};
		\node [style=White] (6) at (2.25, 0) {};
		\node [style=none] (7) at (3.25, 0) {};
	\end{pgfonlayer}
	\begin{pgfonlayer}{edgelayer}
		\draw (0) to (1);
		\draw [bend left] (1) to (2);
		\draw (2) to (4);
		\draw [bend left] (4) to (6);
		\draw (6) to (7.center);
		\draw [bend left] (6) to (5);
		\draw (5) to (3);
		\draw [bend right=330] (3) to (1);
	\end{pgfonlayer}
\end{tikzpicture}

%% file: tikz/figures/rcd_for_pairs_of_matrices/TheoremAbstract/item5/step3.tikz
\begin{tikzpicture}
	\begin{pgfonlayer}{nodelayer}
		\node [style=CoMap] (2) at (-0.75, 0.75) {$D_1$};
		\node [style=CoMap] (3) at (-0.75, -0.75) {$D_2$};
		\node [style=Map] (4) at (0.75, 0.75) {$A$};
		\node [style=Map] (5) at (0.75, -0.75) {$B$};
		\node [style=White] (6) at (2.25, 0) {};
		\node [style=none] (7) at (3.25, 0) {};
		\node [style=Black] (8) at (-2.25, 0.75) {};
		\node [style=Black] (9) at (-2.25, -0.75) {};
	\end{pgfonlayer}
	\begin{pgfonlayer}{edgelayer}
		\draw (2) to (4);
		\draw [bend left] (4) to (6);
		\draw (6) to (7.center);
		\draw [bend left] (6) to (5);
		\draw (5) to (3);
		\draw (8) to (2);
		\draw (9) to (3);
	\end{pgfonlayer}
\end{tikzpicture}

%% file: tikz/figures/rcd_for_pairs_of_matrices/TheoremAbstract/item5/step4.tikz
\begin{tikzpicture}
	\begin{pgfonlayer}{nodelayer}
		\node [style=Black] (2) at (-0.5, 0.75) {};
		\node [style=Black] (3) at (-0.5, -0.75) {};
		\node [style=Map] (4) at (0.75, 0.75) {$A$};
		\node [style=Map] (5) at (0.75, -0.75) {$B$};
		\node [style=White] (6) at (2.25, 0) {};
		\node [style=none] (7) at (3.25, 0) {};
	\end{pgfonlayer}
	\begin{pgfonlayer}{edgelayer}
		\draw (2) to (4);
		\draw [bend left] (4) to (6);
		\draw (6) to (7.center);
		\draw [bend left] (6) to (5);
		\draw (5) to (3);
	\end{pgfonlayer}
\end{tikzpicture}

%% file: tikz/figures/rcd_for_pairs_of_matrices/TheoremAbstract/item2/step1.tikz
\begin{tikzpicture}
	\begin{pgfonlayer}{nodelayer}
		\node [style=none] (7) at (-1.75, 0) {};
		\node [style=Map] (8) at (-0.75, 0) {$D_1$};
		\node [style=Rel] (9) at (0.75, 0) {$H$};
		\node [style=none] (10) at (1.75, 0) {};
	\end{pgfonlayer}
	\begin{pgfonlayer}{edgelayer}
		\draw (7.center) to (8);
		\draw (8) to (9);
		\draw (9) to (10.center);
	\end{pgfonlayer}
\end{tikzpicture}

%% file: tikz/figures/rcd_for_pairs_of_matrices/TheoremAbstract/item2/step2.tikz
\begin{tikzpicture}
	\begin{pgfonlayer}{nodelayer}
		\node [style=White] (1) at (-2.5, 0) {};
		\node [style=CoMap] (2) at (-0.75, 1) {$D_1$};
		\node [style=CoMap] (3) at (-0.75, -1) {$D_2$};
		\node [style=Map] (4) at (0.75, 1) {$A$};
		\node [style=Map] (5) at (0.75, -1) {$B$};
		\node [style=White] (6) at (2.5, 0) {};
		\node [style=none] (7) at (3.5, 0) {};
		\node [style=Map] (8) at (-3.75, 0) {$D_1$};
		\node [style=none] (9) at (-5, 0) {};
	\end{pgfonlayer}
	\begin{pgfonlayer}{edgelayer}
		\draw [bend left] (1) to (2);
		\draw (2) to (4);
		\draw [bend left] (4) to (6);
		\draw (6) to (7.center);
		\draw [bend left] (6) to (5);
		\draw (5) to (3);
		\draw [bend right=330] (3) to (1);
		\draw (9.center) to (8);
		\draw (8) to (1);
	\end{pgfonlayer}
\end{tikzpicture}

%% file: tikz/figures/rcd_for_pairs_of_matrices/TheoremAbstract/item2/step3.tikz
\begin{tikzpicture}
	\begin{pgfonlayer}{nodelayer}
		\node [style=White] (1) at (-2.5, 0) {};
		\node [style=CoMap] (2) at (-0.75, 1) {$D_1$};
		\node [style=White] (3) at (-0.75, -1) {};
		\node [style=Map] (4) at (0.75, 1) {$A$};
		\node [style=White] (5) at (0.75, -1) {};
		\node [style=White] (6) at (2.5, 0) {};
		\node [style=none] (7) at (3.5, 0) {};
		\node [style=Map] (8) at (-3.75, 0) {$D_1$};
		\node [style=none] (9) at (-5, 0) {};
	\end{pgfonlayer}
	\begin{pgfonlayer}{edgelayer}
		\draw [bend left] (1) to (2);
		\draw (2) to (4);
		\draw [bend left] (4) to (6);
		\draw (6) to (7.center);
		\draw [bend left] (6) to (5);
		\draw [bend right=330] (3) to (1);
		\draw (9.center) to (8);
		\draw (8) to (1);
	\end{pgfonlayer}
\end{tikzpicture}

%% file: tikz/figures/rcd_for_pairs_of_matrices/TheoremAbstract/item2/step4.tikz
\begin{tikzpicture}
	\begin{pgfonlayer}{nodelayer}
		\node [style=CoMap] (2) at (-0.75, 0) {$D_1$};
		\node [style=Map] (4) at (0.75, 0) {$A$};
		\node [style=Map] (8) at (-2.25, 0) {$D_1$};
		\node [style=none] (9) at (-3.5, 0) {};
		\node [style=none] (10) at (2, 0) {};
	\end{pgfonlayer}
	\begin{pgfonlayer}{edgelayer}
		\draw (2) to (4);
		\draw (9.center) to (8);
		\draw (8) to (2);
		\draw (4) to (10.center);
	\end{pgfonlayer}
\end{tikzpicture}

%% file: tikz/figures/rcd_for_pairs_of_matrices/Lemma1/lemma.tikz
\begin{tikzpicture}
	\begin{pgfonlayer}{nodelayer}
		\node [style=none] (0) at (0, 0) {$\Leftrightarrow $};
		\node [style=none] (1) at (-0.75, 0) {};
		\node [style=Map] (2) at (-2, 0) {$B$};
		\node [style=none] (3) at (-3.25, 0) {};
		\node [style=none] (4) at (-4, 0) {$\subseteq$};
		\node [style=none] (5) at (-4.75, 0) {};
		\node [style=Map] (6) at (-6, 0) {$A$};
		\node [style=none] (7) at (-7.25, 0) {};
		\node [style=none] (8) at (3.25, 0) {};
		\node [style=Map] (9) at (2, 0) {$A$};
		\node [style=none] (10) at (0.75, 0) {};
		\node [style=none] (11) at (4, 0) {$=$};
		\node [style=none] (12) at (7.25, 0) {};
		\node [style=Map] (13) at (6, 0) {$B$};
		\node [style=none] (14) at (4.75, 0) {};
	\end{pgfonlayer}
	\begin{pgfonlayer}{edgelayer}
		\draw (6) to (5.center);
		\draw (7.center) to (6);
		\draw (3.center) to (2);
		\draw (2) to (1.center);
		\draw (9) to (8.center);
		\draw (10.center) to (9);
		\draw (14.center) to (13);
		\draw (13) to (12.center);
	\end{pgfonlayer}
\end{tikzpicture}

%% file: tikz/figures/rcd_for_pairs_of_matrices/Lemma1/proof1.tikz
\begin{tikzpicture}
	\begin{pgfonlayer}{nodelayer}
		\node [style=none] (1) at (-0.75, 0) {};
		\node [style=Map] (2) at (-2, 0) {$B$};
		\node [style=CoMap] (4) at (-3.5, 0) {$A$};
		\node [style=Map] (6) at (-5, 0) {$A$};
		\node [style=none] (7) at (-6.25, 0) {};
	\end{pgfonlayer}
	\begin{pgfonlayer}{edgelayer}
		\draw (7.center) to (6);
		\draw (2) to (1.center);
		\draw (6) to (4);
		\draw (4) to (2);
	\end{pgfonlayer}
\end{tikzpicture}

%% file: tikz/figures/rcd_for_pairs_of_matrices/Lemma1/proof2.tikz
\begin{tikzpicture}
	\begin{pgfonlayer}{nodelayer}
		\node [style=none] (1) at (-0.75, 0) {};
		\node [style=Map] (2) at (-2, 0) {$B$};
		\node [style=CoMap] (4) at (-3.5, 0) {$B$};
		\node [style=Map] (6) at (-5, 0) {$A$};
		\node [style=none] (7) at (-6.25, 0) {};
	\end{pgfonlayer}
	\begin{pgfonlayer}{edgelayer}
		\draw (7.center) to (6);
		\draw (2) to (1.center);
		\draw (6) to (4);
		\draw (4) to (2);
	\end{pgfonlayer}
\end{tikzpicture}

%% file: tikz/figures/rcd_for_cospan/lemma_break/lemma_hyp1.tikz
\begin{tikzpicture}
	\begin{pgfonlayer}{nodelayer}
		\node [style=White] (0) at (0.5, 0.25) {};
		\node [style=none] (1) at (0.5, -0.5) {};
		\node [style=Map] (2) at (2.5, 0) {$A$};
		\node [style=none] (3) at (2.5, 0.25) {};
		\node [style=none] (4) at (2.5, -0.5) {};
		\node [style=none] (5) at (4.5, 0) {};
		\node [style=none] (6) at (1.25, 0.5) {$\scriptstyle{m_1}$};
		\node [style=none] (7) at (1.25, -0.25) {$\scriptstyle{m_2}$};
		\node [style=none] (8) at (4, 0.25) {$\scriptstyle{n}$};
	\end{pgfonlayer}
	\begin{pgfonlayer}{edgelayer}
		\draw (0) to (3.center);
		\draw (1.center) to (4.center);
		\draw (2) to (5.center);
	\end{pgfonlayer}
\end{tikzpicture}

%% file: tikz/figures/rcd_for_cospan/lemma_break/lemma_hyp2.tikz
\begin{tikzpicture}
	\begin{pgfonlayer}{nodelayer}
		\node [style=none] (0) at (0.75, 0) {};
		\node [style=Map] (1) at (2.5, 0) {$A_2$};
		\node [style=none] (2) at (4.25, 0) {};
		\node [style=none] (3) at (1.25, 0.25) {$\scriptstyle{m_2}$};
		\node [style=none] (4) at (3.75, 0.25) {$\scriptstyle{n}$};
	\end{pgfonlayer}
	\begin{pgfonlayer}{edgelayer}
		\draw (0.center) to (1);
		\draw (1) to (2.center);
	\end{pgfonlayer}
\end{tikzpicture}

%% file: tikz/figures/rcd_for_cospan/lemma_break/step1.tikz
\begin{tikzpicture}
	\begin{pgfonlayer}{nodelayer}
		\node [style=White] (0) at (0.25, 1) {};
		\node [style=none] (1) at (0.25, -1) {};
		\node [style=White] (2) at (4.5, 0) {};
		\node [style=Map] (3) at (2.5, 1) {$A_1$};
		\node [style=Map] (4) at (2.5, -1) {$A_2$};
		\node [style=none] (5) at (6.25, 0) {};
		\node [style=none] (6) at (1, 1.25) {$\scriptstyle{m_1}$};
		\node [style=none] (7) at (1, -0.75) {$\scriptstyle{m_2}$};
		\node [style=none] (8) at (5.75, 0.25) {$\scriptstyle{n}$};
	\end{pgfonlayer}
	\begin{pgfonlayer}{edgelayer}
		\draw (0) to (3);
		\draw (1.center) to (4);
		\draw (2) to (5.center);
		\draw [bend left=15] (3) to (2);
		\draw [bend right=15] (4) to (2);
	\end{pgfonlayer}
\end{tikzpicture}

%% file: tikz/figures/rcd_for_cospan/lemma_break/step2.tikz
\begin{tikzpicture}
	\begin{pgfonlayer}{nodelayer}
		\node [style=White] (0) at (2.25, 1) {};
		\node [style=none] (1) at (0.25, -1) {};
		\node [style=White] (2) at (4.5, 0) {};
		\node [style=none] (3) at (2.5, 1) {};
		\node [style=Map] (4) at (2.5, -1) {$A_2$};
		\node [style=none] (5) at (6.25, 0) {};
		\node [style=none] (6) at (3, 1.25) {$\scriptstyle{m_1}$};
		\node [style=none] (7) at (1, -0.75) {$\scriptstyle{m_2}$};
		\node [style=none] (8) at (5.75, 0.25) {$\scriptstyle{n}$};
	\end{pgfonlayer}
	\begin{pgfonlayer}{edgelayer}
		\draw (0) to (3.center);
		\draw (1.center) to (4);
		\draw (2) to (5.center);
		\draw [bend left=15] (3.center) to (2);
		\draw [bend right=15] (4) to (2);
	\end{pgfonlayer}
\end{tikzpicture}

%% file: tikz/figures/rcd_for_cospan/cospan/subspace.tikz
\begin{tikzpicture}
	\begin{pgfonlayer}{nodelayer}
		\node [style=Rel] (0) at (-0.75, 0) {$V$};
		\node [style=none] (1) at (0.75, 0) {};
	\end{pgfonlayer}
	\begin{pgfonlayer}{edgelayer}
		\draw (0) to (1.center);
	\end{pgfonlayer}
\end{tikzpicture}

%% file: tikz/figures/rcd_for_cospan/cospan/subspace-image.tikz
\begin{tikzpicture}
	\begin{pgfonlayer}{nodelayer}
		\node [style=Rel] (0) at (-2.25, 0) {$V$};
		\node [style=none] (1) at (-0.75, 0) {};
		\node [style=none] (3) at (0, 0) {$=$};
		\node [style=Black] (4) at (0.75, 0) {};
		\node [style=Map] (5) at (2.25, 0) {$X_1$};
		\node [style=none] (6) at (3.75, 0) {};
	\end{pgfonlayer}
	\begin{pgfonlayer}{edgelayer}
		\draw (0) to (1.center);
		\draw (4) to (5);
		\draw (5) to (6.center);
	\end{pgfonlayer}
\end{tikzpicture}

%% file: tikz/figures/rcd_for_cospan/cospan/subspace-kernel.tikz
\begin{tikzpicture}
	\begin{pgfonlayer}{nodelayer}
		\node [style=Rel] (0) at (-2.25, 0) {$V$};
		\node [style=none] (1) at (-0.75, 0) {};
		\node [style=none] (3) at (0, 0) {$=$};
		\node [style=White] (4) at (0.75, 0) {};
		\node [style=CoMap] (5) at (2.25, 0) {$X_2$};
		\node [style=none] (6) at (3.75, 0) {};
	\end{pgfonlayer}
	\begin{pgfonlayer}{edgelayer}
		\draw (0) to (1.center);
		\draw (4) to (5);
		\draw (5) to (6.center);
	\end{pgfonlayer}
\end{tikzpicture}

%% file: tikz/figures/rcd_for_cospan/cospan/lemmacospan/step1.tikz
\begin{tikzpicture}
	\begin{pgfonlayer}{nodelayer}
		\node [style=Rel] (0) at (0, 1.5) {$R$};
		\node [style=White] (1) at (1.75, 0.5) {};
		\node [style=none] (2) at (-1.75, 1.5) {};
		\node [style=White] (3) at (3.25, 0.5) {};
		\node [style=none] (4) at (0, -0.5) {};
		\node [style=White] (5) at (-1.75, -1.5) {};
		\node [style=White] (6) at (-3.25, -1.5) {};
		\node [style=none] (7) at (0, -2.5) {};
		\node [style=none] (8) at (2, -2.5) {};
		\node [style=none] (9) at (-0.75, 2.5) {};
		\node [style=none] (10) at (-0.75, -1) {};
		\node [style=none] (11) at (3.75, 2.5) {};
		\node [style=none] (12) at (3.75, -1) {};
	\end{pgfonlayer}
	\begin{pgfonlayer}{edgelayer}
		\draw [bend left] (0) to (1);
		\draw (2.center) to (0);
		\draw (1) to (3);
		\draw [bend left] (1) to (4.center);
		\draw [bend left=330] (4.center) to (5);
		\draw [bend right] (5) to (7.center);
		\draw (7.center) to (8.center);
		\draw (5) to (6);
		\draw [style=new edge style 3] (9.center) to (10.center);
		\draw [style=new edge style 3] (9.center) to (11.center);
		\draw [style=new edge style 3] (11.center) to (12.center);
		\draw [style=new edge style 3] (10.center) to (12.center);
	\end{pgfonlayer}
\end{tikzpicture}

%% file: tikz/figures/rcd_for_cospan/cospan/lemmacospan/step2.tikz
\begin{tikzpicture}
	\begin{pgfonlayer}{nodelayer}
		\node [style=White] (1) at (2, 1) {};
		\node [style=Map] (4) at (0.25, 1) {$X$};
		\node [style=White] (5) at (-1.75, 0) {};
		\node [style=White] (6) at (-3.25, 0) {};
		\node [style=none] (7) at (0, -1) {};
		\node [style=none] (8) at (2, -1) {};
		\node [style=none] (9) at (-2, 1.25) {};
		\node [style=none] (10) at (0.25, 1.25) {};
		\node [style=none] (11) at (0.25, 0.75) {};
	\end{pgfonlayer}
	\begin{pgfonlayer}{edgelayer}
		\draw (1) to (4);
		\draw [bend right] (5) to (7.center);
		\draw (7.center) to (8.center);
		\draw (5) to (6);
		\draw (9.center) to (10.center);
		\draw [bend left] (5) to (11.center);
	\end{pgfonlayer}
\end{tikzpicture}

%% file: tikz/figures/rcd_for_cospan/cospan/lemmacospan/step3.tikz
\begin{tikzpicture}
	\begin{pgfonlayer}{nodelayer}
		\node [style=White] (1) at (2, 0.75) {};
		\node [style=White] (5) at (-1.75, -0.75) {};
		\node [style=White] (6) at (-3.25, -0.75) {};
		\node [style=none] (7) at (0, -1.75) {};
		\node [style=none] (8) at (2, -1.75) {};
		\node [style=none] (9) at (-2, 1.5) {};
		\node [style=Map] (10) at (0.25, 1.5) {$A$};
		\node [style=Map] (11) at (0.25, 0) {$B$};
		\node [style=White] (12) at (3.5, 0.75) {};
	\end{pgfonlayer}
	\begin{pgfonlayer}{edgelayer}
		\draw [bend right] (5) to (7.center);
		\draw (7.center) to (8.center);
		\draw (5) to (6);
		\draw (9.center) to (10);
		\draw [bend left] (5) to (11);
		\draw [bend left, looseness=0.75] (10) to (1);
		\draw [bend right, looseness=0.75] (11) to (1);
		\draw (1) to (12);
	\end{pgfonlayer}
\end{tikzpicture}

%% file: tikz/figures/rcd_for_cospan/cospan/lemmacospan/step3-1.tikz
\begin{tikzpicture}
	\begin{pgfonlayer}{nodelayer}
		\node [style=White] (1) at (2, 0.75) {};
		\node [style=White] (5) at (-1.75, -0.75) {};
		\node [style=White] (6) at (-3.25, -0.75) {};
		\node [style=none] (7) at (0, -1.75) {};
		\node [style=none] (8) at (2, -1.75) {};
		\node [style=none] (9) at (-2, 1.5) {};
		\node [style=Map] (10) at (0.25, 1.5) {$A$};
		\node [style=Map] (11) at (0.25, 0) {$B$};
		\node [style=CoMap] (12) at (3.25, 0.75) {$B$};
		\node [style=White] (13) at (4.5, 0.75) {};
	\end{pgfonlayer}
	\begin{pgfonlayer}{edgelayer}
		\draw [bend right] (5) to (7.center);
		\draw (7.center) to (8.center);
		\draw (5) to (6);
		\draw (9.center) to (10);
		\draw [bend left] (5) to (11);
		\draw [bend left, looseness=0.75] (10) to (1);
		\draw [bend right, looseness=0.75] (11) to (1);
		\draw (1) to (12);
		\draw [in=180, out=0] (12) to (13);
	\end{pgfonlayer}
\end{tikzpicture}

%% file: tikz/figures/rcd_for_cospan/cospan/lemmacospan/step4.tikz
\begin{tikzpicture}
	\begin{pgfonlayer}{nodelayer}
		\node [style=White] (1) at (2, 0.75) {};
		\node [style=White] (5) at (-1.75, -0.75) {};
		\node [style=White] (6) at (-3.25, -0.75) {};
		\node [style=none] (7) at (0, -1.75) {};
		\node [style=none] (8) at (2, -1.75) {};
		\node [style=none] (9) at (-2.75, 1.5) {};
		\node [style=Map] (10) at (-1.25, 1.5) {$A$};
		\node [style=none] (11) at (0.25, 0) {};
		\node [style=White] (12) at (3.5, 0.75) {};
		\node [style=CoMap] (13) at (0.25, 1.5) {$B$};
	\end{pgfonlayer}
	\begin{pgfonlayer}{edgelayer}
		\draw [bend right] (5) to (7.center);
		\draw (7.center) to (8.center);
		\draw (5) to (6);
		\draw (9.center) to (10);
		\draw [bend left] (5) to (11.center);
		\draw [bend right] (11.center) to (1);
		\draw (1) to (12);
		\draw (10) to (13);
		\draw [bend left] (13) to (1);
	\end{pgfonlayer}
\end{tikzpicture}

%% file: tikz/figures/rcd_for_cospan/cospan/lemmacospan/step5.tikz
\begin{tikzpicture}
	\begin{pgfonlayer}{nodelayer}
		\node [style=none] (1) at (2.25, 0) {};
		\node [style=none] (9) at (-2.25, 0) {};
		\node [style=Map] (10) at (-0.75, 0) {$A$};
		\node [style=CoMap] (13) at (0.75, 0) {$B$};
	\end{pgfonlayer}
	\begin{pgfonlayer}{edgelayer}
		\draw (9.center) to (10);
		\draw (10) to (13);
		\draw (13) to (1.center);
	\end{pgfonlayer}
\end{tikzpicture}

%% file: main.bbl
\begin{thebibliography}{10}
\providecommand{\url}[1]{\texttt{#1}}
\providecommand{\urlprefix}{URL }
\expandafter\ifx\csname urlstyle\endcsname\relax
  \providecommand{\doi}[1]{doi:\discretionary{}{}{}#1}\else
  \providecommand{\doi}{doi:\discretionary{}{}{}\begingroup
  \urlstyle{rm}\Url}\fi
\providecommand{\eprint}[2][]{\url{#2}}

\bibitem{weisstein2014invertible}
Weisstein EW.
\newblock Invertible matrix theorem.
\newblock \emph{https://mathworld.wolfram.com/}, 2014.

\bibitem{axler2024linear}
Axler SJ.
\newblock Linear Algebra Done Right.
\newblock Springer International Publishing, 4 edition, 2024.
\newblock ISBN 9783031410260.
\newblock \doi{10.1007/978-3-031-41026-0}.
\newblock \urlprefix\url{http://dx.doi.org/10.1007/978-3-031-41026-0}.

\bibitem{hinze2023introducing}
Hinze R, Marsden D.
\newblock Introducing String Diagrams: The Art of Category Theory.
\newblock Cambridge University Press, 2023.

\bibitem{ZanasiThesis}
Zanasi F.
\newblock Interacting {H}opf Algebras: the theory of linear systems.
\newblock Ph.D. thesis, {E}cole Normale Sup\'{e}rieure de Lyon, 2015.

\bibitem{bonchi2014categorical}
Bonchi F, Soboci{\'n}ski P, Zanasi F.
\newblock A categorical semantics of signal flow graphs.
\newblock In: International Conference on Concurrency Theory. Springer, 2014
  pp. 435--450.

\bibitem{Bonchi2015}
Bonchi F, Soboci\'{n}ski P, Zanasi F.
\newblock Full Abstraction for Signal Flow Graphs.
\newblock In: POPL 2015. ACM, 2015 pp. 515--526.

\bibitem{bonchi2017refinement}
Bonchi F, Holland J, Pavlovic D, Soboci\'nski P.
\newblock {Refinement for Signal Flow Graphs}.
\newblock In: Meyer R, Nestmann U (eds.), 28th International Conference on
  Concurrency Theory (CONCUR 2017), volume~85 of \emph{Leibniz International
  Proceedings in Informatics (LIPIcs)}. Schloss Dagstuhl -- Leibniz-Zentrum
  f{\"u}r Informatik, Dagstuhl, Germany.
\newblock ISBN 978-3-95977-048-4, 2017 pp. 24:1--24:16.
\newblock \doi{10.4230/LIPIcs.CONCUR.2017.24}.
\newblock
  \urlprefix\url{https://drops.dagstuhl.de/entities/document/10.4230/LIPIcs.CONCUR.2017.24}.

\bibitem{Bonchi2019a}
Bonchi F, Piedeleu R, Sobociński P, Zanasi F.
\newblock Graphical Affine Algebra.
\newblock In: 2019 34th Annual ACM/IEEE Symposium on Logic in Computer Science
  (LICS). 2019 pp. 1--12.
\newblock \doi{10.1109/LICS.2019.8785877}.

\bibitem{PAIXAO2022}
Paixão J, Rufino L, Soboci\'{n}ski P.
\newblock High-level axioms for graphical linear algebra.
\newblock \emph{Science of Computer Programming}, 2022.
\newblock \textbf{218}:102791.
\newblock \doi{https://doi.org/10.1016/j.scico.2022.102791}.
\newblock
  \urlprefix\url{https://www.sciencedirect.com/science/article/pii/S0167642322000247}.

\bibitem{blogpawel}
Soboci\'{n}ski P.
\newblock Graphical Linear Algebra, “Maths with Diagrams”, Episode 17.
\newblock
  \url{https://graphicallinearalgebra.net/2015/06/16/maths-with-diagrams/},
  2015.
\newblock Accessed: 03-08-2023.

\bibitem{Dijkstra1989PredicateCA}
Dijkstra EW, Scholten CS.
\newblock Predicate Calculus and Program Semantics.
\newblock Monographs in Theoretical Computer Science. Springer, 1989.
\newblock ISBN 978-0387507495.

\bibitem{gsvd1976}
Van~Loan CF.
\newblock Generalizing the Singular Value Decomposition.
\newblock \emph{SIAM Journal on Numerical Analysis}, 1976.
\newblock \textbf{13}(1):76--83.
\newblock \doi{10.1137/0713009}.
\newblock \eprint{https://doi.org/10.1137/0713009},
  \urlprefix\url{https://doi.org/10.1137/0713009}.

\bibitem{gsvd1981}
Paige CC, Saunders MA.
\newblock Towards a Generalized Singular Value Decomposition.
\newblock \emph{SIAM Journal on Numerical Analysis}, 1981.
\newblock \textbf{18}(3):398--405.
\newblock \doi{10.1137/0718026}.
\newblock \eprint{https://doi.org/10.1137/0718026},
  \urlprefix\url{https://doi.org/10.1137/0718026}.

\bibitem{booth2024completeequationaltheoriesclassical}
Booth RI, Carette T, Comfort C.
\newblock Complete equational theories for classical and quantum Gaussian
  relations, 2024.
\newblock \eprint{2403.10479},
  \urlprefix\url{https://arxiv.org/abs/2403.10479}.

\bibitem{freyd1990categories}
Freyd P, Scedrov A.
\newblock Categories, Allegories, volume~39.
\newblock North-Holland, Amsterdam, 1990.
\newblock ISBN 9780444548856, 9780444703675, 9780444703682, 9780080887012.
\newblock Paperback ISBN: 9780444548856, Paperback ISBN: 9780444703675,
  Hardback ISBN: 9780444703682, eBook ISBN: 9780080887012.

\bibitem{Fischer2012}
Fischer G.
\newblock Lernbuch Lineare Algebra Und Analytische Geometrie.
\newblock Springer Vieweg, 2 edition, 2012.

\bibitem{selinger2010survey}
Selinger P.
\newblock A survey of graphical languages for monoidal categories.
\newblock In: New structures for physics, pp. 289--355. Springer, 2010.

\end{thebibliography}
